\documentclass[11pt,letterpaper]{article}
\usepackage[margin=1in]{geometry}
\usepackage[utf8]{inputenc}

\newcommand{\D}{\displaystyle}

\newcommand{\classNP}{\textsf{NP}}

\newcommand{\classSharpP}{\textsf{\#P}}

\newcommand{\bfIf}{\textbf{if}}
\newcommand{\bfThen}{\textbf{then}}
\newcommand{\bfFor}{\textbf{for}}
\newcommand{\bfDo}{\textbf{do}}

\usepackage{amsmath,amssymb,amsthm}
\usepackage{xcolor}
\usepackage[hidelinks]{hyperref}
\usepackage[ruled,vlined,linesnumbered]{algorithm2e}
\usepackage{bm}
\usepackage{nicefrac}
\usepackage{paralist}
\usepackage{typearea}

\typearea{16}

\newtheorem{definition}{Definition}
\newtheorem{lemma}{Lemma}
\newtheorem{theorem}{Theorem}
\newtheorem{proposition}{Proposition}

\newtheorem{corollary}{Corollary}

\newcommand{\dist}{\ensuremath{\mathcal{D}}}
\newcommand{\sv}{\xi}
\newcommand{\rv}{\rho}

\newcommand{\vecx}{\ensuremath{\bm{x}}}

\newcommand{\vecy}{\ensuremath{\bm{y}}}
\newcommand{\vecz}{\ensuremath{\bm{z}}}

\newcommand{\vecAlpha}{\ensuremath{\bm{\alpha}}}

\newcommand{\calP}{\mathcal{P}}
\newcommand{\calS}{\mathcal{S}}

\newcommand{\calR}{\mathcal{R}}

\newcommand{\receiver}{\ensuremath{\calR}}
\newcommand{\sender}{\ensuremath{\calS}}
\newcommand{\vecState}{\ensuremath{\boldsymbol{\theta}}}
\newcommand{\state}{\ensuremath{\theta}}
\newcommand{\opt}{\operatorname{OPT}}
\newcommand{\Ex}[1]{\ensuremath{\mathbb{E}\left[#1\right]}}

\renewcommand{\rho}{\varrho}
\renewcommand{\phi}{\varphi}

\newcommand{\growingmid}{\mathrel{}\middle|\mathrel{}}

\allowdisplaybreaks

\title{Algorithms for Persuasion with Limited Communication}
\author{%
Ronen Gradwohl%
\thanks{Department of Economics and Business Administration, Ariel University, Israel. Email: roneng@ariel.ac.il.
	Gradwohl gratefully acknowledges the support of National Science Foundation award number 1718670.}
\and
Niklas Hahn%
\thanks{Institute of Computer Science, Goethe University Frankfurt, Germany. Email: nhahn@em.uni-frankfurt.de. Hahn gratefully acknowledges the support of German-Israeli Foundation grant I-1419-118.4/2017.}
\and Martin Hoefer%
\thanks{Institute of Computer Science, Goethe University Frankfurt, Germany. Email: mhoefer@em.uni-frankfurt.de. Hoefer gratefully acknowledges the support of German-Israeli Foundation grant I-1419-118.4/2017 and Deutsche Forschungsgemeinschaft grants DFG Ho 3831/5-1, 6-1, and 7-1.}
\and Rann Smorodinsky%
\thanks{Faculty of Industrial Engineering and Management, Technion, Israel.
	Email: rann@ie.technion.ac.il. Smorodinsky gratefully acknowledges support of United States-Israel Binational
	Science Foundation and National Science Foundation grant 2016734, German-Israeli Foundation grant
	I-1419-118.4/2017, Ministry of Science and Technology grant 19400214, Technion VPR grants, and the
	Bernard M. Gordon Center for Systems Engineering at the Technion.}
}
\date{}

\begin{document}
\maketitle
\newcommand{\alert}[1]{{\color{red} #1}}

\begin{abstract}
The Bayesian persuasion paradigm of strategic communication models interaction between a pri\-vately-informed agent, called the \emph{sender}, and an ignorant but rational agent, called the \emph{receiver}. The goal is typically to design a (near-)optimal communication (or \emph{signaling}) scheme for the sender. It enables the sender to disclose information to the receiver in a way as to incentivize her to take an action that is preferred by the sender. Finding the optimal signaling scheme is known to be computationally difficult in general. This hardness is further exacerbated when there is also a constraint on the size of the message space, leading to \classNP-hardness of approximating the optimal sender utility within any constant factor.

In this paper, we show that in several natural and prominent cases the optimization problem is tractable even when the message space is limited. In particular, we study signaling under a symmetry or an independence assumption on the distribution of utility values for the actions. For symmetric distributions, we provide a novel characterization of the optimal signaling scheme. It results in a polynomial-time algorithm to compute an optimal scheme for many compactly represented symmetric distributions. In the independent case, we design a constant-factor approximation algorithm, which stands in marked contrast to the hardness of approximation in the general case.

\end{abstract}

\pagenumbering{gobble}
\clearpage
\pagenumbering{arabic}


\section{Introduction}
Recommendations play a vital role in the modern information economy: Online retailers make product recommendations, travel websites provide advice on
hotels and attractions, navigation apps suggest driving routes, and so on. In all these examples, the designers of the recommendation systems have information that
consumers do not, and both sides benefit from communication. However, the interests of the consumers and the
recommenders are not always aligned. For example, while consumers may prefer to purchase products that
constitute a better bargain, retailers may prefer to sell products for which they obtain higher margins. A natural goal is to optimize the use of the retailer's informational advantage, such that recommendations result in consumer choices that maximize its own benefit. In doing so, the retailer must account for the consumers' interests to guarantee that recommendations are being followed.

This optimization problem fits into the Bayesian persuasion paradigm of \cite{KamenicaG11},
a fundamental model of strategic communication proposed in economics that has recently gained significant interest in algorithmic game theory. In this model there are two players: a {\em sender} $\sender$ with information about a so-called {\em state of nature}, and a {\em receiver} $\receiver$ who takes an action.
Payoffs of the two players are determined both by the action chosen by $\receiver$ and by the state of nature.
A priori, the players do not know the true state of nature, but rather share a common belief (i.e., a distribution) over the possible outcomes. However, $\sender$ obtains information about the \emph{realized} state of nature, and then sends a message (called a \emph{signal}) to $\receiver$. After receiving the signal, $\receiver$ takes an action,
and payoffs are realized.

A distinguishing feature of the Bayesian persuasion model is that $\sender$ commits to a \emph{signaling scheme} before the state of nature is realized. A signaling scheme is a (possibly randomized) function from states of nature to signals. The action for $\sender$ can be cast as choosing a signaling scheme that determines the signal once the state of nature is realized. This problem becomes interesting, above and beyond a standard optimization, when $\sender$ and $\receiver$ have misaligned preferences with different optimal actions in various states. How can $\sender$ make optimal use of her informational advantage in steering $\receiver$'s choice of action?


The problem of optimally designing recommendation systems fits neatly into this model.
To illustrate, consider the following simple example: $\sender$ is a retailer that makes a product recommendation to a consumer, $\receiver$, who must choose one of the products. The various products yield different utilities to each of the players and, while $\sender$ knows which product yields which utilities,
from $\receiver$'s perspective the products are randomly ordered. The state of nature is the order in which the products appear, and the signaling scheme is the recommendation system implemented by the retailer.

To make the example concrete, suppose there are three products: One product is good for $\sender$ and bad for $\receiver$, one is bad
for $\sender$ but good for $\receiver$, and one is bad for both. Denote these respective products by $GB$, $BG$, and $BB$, and suppose they yield sender-receiver utility pairs $(1,0)$, $(0,1)$, and $(0,0)$, when chosen.
One signaling scheme for the sender is to always reveal which product is which. In this case $\receiver$ will choose $BG$, and $\sender$ will attain utility 0.
A better scheme for $\sender$ is to reveal no information. Here the best $\receiver$ can do is choose randomly, in which case $\sender$'s utility will be $1/3$.
One might attempt to improve $\sender$'s utility by always recommending $GB$. However, this policy is not \emph{persuasive}: $\receiver$'s optimal reaction is to deviate to choosing one of the \emph{other} two products at random, and again $\sender$'s utility will be 0. Nonetheless, $\sender$ can do better than the no-information scheme by choosing a scheme that recommends $GB$ with probability $2/3$ and $BG$ with probability $1/3$. A straightforward calculation using Bayes' Rule shows that $\receiver$ cannot improve by deviating from this recommendation, and that following it leads to sender utility of $2/3$. This, in fact, is the optimal signaling scheme for $\sender$.

In this paper we study potential barriers to optimal signaling, focusing on two constraints: limited communication and limited computational resources. First, in our example above, the optimal signaling scheme needs a signal space of size 3, as each of the three products could potentially be $GB$ or $BG$. But what if she was restricted to sending only one of 2 signals? More generally, suppose there are $n$ products, but $\sender$ is restricted to only $k$ signals. These restrictions arise naturally, e.g., when there is a limited attention span, or communication between the players is noisy and a limited number of bits can be transferred. Typically, designing optimal signaling schemes can be based on the popular toolset developed by \cite{KamenicaG11}. However, these tools no longer apply when the number of available signals is limited.

Second, from a computational perspective, finding the optimal scheme might not be tractable. Suppose in the example above there are $n$ products, for large $n$. For the restricted case in which the utility-pairs of the $n$ products are IID and given explicitly, Dughmi and Xu~\cite{DughmiX16} develop a polynomial-time algorithm that computes the optimal scheme. Note that our example above, in which the utility-pairs are known but their order is not, does not fall into this case.
On the other hand, for general distributions over the utility-pairs of each product (and even ones that are independent but not identical), they
show that computing the optimal sender utility is \classSharpP-hard~\cite{DughmiX16}.

Third, when computational concerns are combined with limited communication, the computational problem is exacerbated. Dughmi, Xu and Qiang~\cite{DughmiKQ16} prove a substantially stronger hardness result and show that it is \classNP-hard to even approximate the optimal sender utility to within any constant factor.

\paragraph{Results and Contribution}
We analyze optimal signaling schemes subject to communication and computation limits in the context of two specific classes of
problems that we call symmetric instances and independent instances. Symmetric instances are ones in which the a priori probability of any vector of $n$ utility-pairs is the same as the a priori probability of any vector in which the $n$ elements have been permuted.
For example, we described a symmetric instance in our example above, in which the $n$ products appear in a random order.
Another example is the IID case, in which products' utility-pairs are drawn i.i.d.\ from a single distribution. The class is more general than
both examples; in Section~\ref{sec:model} we describe some other cases that it captures.

In Section~\ref{sec:symmetric} we study the class of symmetric instances and develop a geometric characterization of the optimal signaling scheme. In Section~\ref{sec:symmetricCompute} we use this characterization to design an algorithm that computes optimal schemes. Our algorithm runs in polynomial time given access to a \emph{probability oracle} that computes certain probabilities related to the instance. We then prove that the probability oracle can be implemented in polynomial time in many prominent subclasses of instances studied in related literature, including but not limited to the IID and random-order cases. Our results significantly expand the set of instances for which optimal schemes can be computed efficiently beyond the IID case in~\cite{DughmiX16}.

Interestingly, our results extend even to limited signal spaces. In addition to the geometric characterization of optimal signaling schemes with limited communication, our results also imply a polynomial-time algorithm for finding such a scheme. Moreover, when relaxing the persuasiveness constraint, we show that a bicriteria approximation can be obtained in polynomial time, see Appendix~\ref{app:bicriteria}.

In Section~\ref{sec:independent} we develop polynomial-time algorithms for finding an approximately-optimal signaling scheme in a class of independent instances in which the utility-pairs are independently but not identically distributed among the $n$ actions. For general independent instances, \cite{DughmiX16} show that finding an optimal solution is \classSharpP-hard. We obtain a constant-factor approximation when the optimal scheme must guarantee for every signal at least the best a-priori utility of any action for the receiver. This is the case, e.g., when an action with a-priori best utility for the receiver has deterministic utility for the receiver. Alternatively, this is the case, when the receiver has an outside option of a-priori optimal utility. Our first algorithm in Section~\ref{sec:constantFactor} is simple to state and implement and guarantees a constant-factor approximation, even in the case in which the signal space is restricted to $k<n$ signals. The ratio is at least 0.375 for $k=2$, and it approaches $(1-1/e)^2 \approx 0.3996$ for large $k$. With a significantly more elaborate procedure in Section~\ref{sec:FPTAS}, we improve the approximation ratio for large $k$ to $(1-1/e-\varepsilon) \approx 0.632$, for any constant $\varepsilon > 0$. With the techniques used here it is impossible to obtain a better ratio than $1-1/e$.

These results stand in marked contrast to the hardness result of \cite{DughmiKQ16} for general instances, where restrictions on the signaling space can make the optimization problem hard to approximate within any constant factor. Our results significantly broaden the class of instances for which good approximation algorithms are known to exist.

Finally, in Section~\ref{sec:apxNbyK} we show that restricting the number of signals from $n$ to $k$ hurts the optimal sender utility by a (tight) factor of $\Theta(k/n)$ in symmetric instances and in independent instances.

\paragraph{Techniques}
Our main results on symmetric instances in Sections~\ref{sec:symmetric} and~\ref{sec:symmetricCompute} use a geometric characterization of the optimal signaling scheme. For every state of nature, we interpret the utility pairs of the $n$ actions as a set of points in the two-dimensional plane. Given a state of nature, the expected utility of any signaling scheme can be interpreted as a \emph{recommendation point} inside the convex hull of the point set. We show that the optimal scheme has a symmetry property and, for every state of nature, its recommendation point is located on the Pareto frontier of the point set. Their location is such that a single common slope lies tangent to the recommendation point for every state of nature.
The symmetry property allows to tightly capture the persuasiveness constraint as a linear inequality. Using these insights, we turn the computation into solving a polynomial number of linear programs. The coefficients are probabilities derived from the Pareto frontiers of point sets of the states of nature. In this way, computing the optimal signaling scheme reduces to computing certain probabilities. We show that for a variety of symmetric distributions, such as IID, random-order, prophet-secretary, or explicitly represented ones (for formal definitions see Section~\ref{sec:model} below) computing these probabilities can be done in polynomial time.

Our results provide an alternative way to compute an optimal scheme for the IID case. The previous approach of~\cite{DughmiX16} uses symmetry to apply techniques from the literature on designing optimal auctions with money. These techniques crucially rely on independence among bidders/actions. In contrast, our characterization and algorithms directly exploit the structure of the persuasion problem. We can handle correlations in the utility-pairs of the state of nature and obtain efficient algorithms for symmetric instances in full generality, even in the case with limited communication.

Our approximation algorithms for independent instances in Section~\ref{sec:independent} follow a two-step approach: (a) find a good subset of $k$ actions and (b) use each of the $k$ signals to recommend one action from the subset. By dropping and relaxing some constraints of the optimal signaling scheme, we devise an LP relaxation. For this relaxation, we prove that step (a) becomes a submodular optimization problem, for which we use the standard greedy algorithm. For (b) we develop an algorithm turning the optimal solution of the LP relaxation into a persuasive signaling scheme. This algorithm in Section~\ref{sec:constantFactor} yields an approximation ratio of roughly $(1-1/e)$ (for large $k$) in each of these steps. Our improved analysis in Section~\ref{sec:FPTAS} then shows that for large $k$ the greedy algorithm for step (a) can be replaced by an FPTAS, but the factor $1-1/e$ from step (b) remains. The latter factor turns out to be tight -- a further improvement must bypass the use of the LP relaxation to upper bound the optimal sender utility.

\paragraph{Related Literature}
Originating in Aumann and Maschler's~\cite{AumannM66} work on repeated games with incomplete information, Bayesian persuasion was popularized by Kamenica and Gentzkow~\cite{KamenicaG11}. The many applications include financial-sector stress testing \cite{GoldsteinL18}, medical research \cite{Kolotilin15},  security \cite{RabinovichJJX15,XuRDT15,XuFCDT16}, online advertisement \cite{ArieliB19,BadanidiyuruBX18,EmekFGLT12} and voting \cite{AlonsoC16}---thorough overviews include \cite{BergemannM19,Dughmi17,Kamenica19,Forges20}.

Our paper analyzes algorithmic Bayesian persuasion with limited signal spaces, most closely related to Dughmi and Xu~\cite{DughmiX16} and Dughmi et al.~\cite{DughmiKQ16}.
The former give a poly-time algorithm to calculate the optimal scheme for IID instances, and show that the problem is \#P-hard in the independently- but not identically-distributed case.
The latter
focus on bilateral trade with constrained communication, but prove two general results: (i) only a $O(\tiny{\frac{\#Signals}{\#States}})$ factor of utility in the unconstrained communication scenario is obtainable by the sender, and (ii) it is NP-hard to approximate the optimal sender utility within a constant factor with a limited number of signals. Our work complements these,
as we give an optimal polynomial-time algorithm for symmetric instances, and  a polynomial-time constant-factor approximation for the class of independent instances.


Another related paper, Aybas and Turkel~\cite{AybasT20}, proves the existence of an optimal scheme when signals are limited. They also show that the sender loses at most a $2/k$
factor of utility when the number of signals decreases from $k$ to $k-1$. We strengthen this result for symmetric instances by showing that the cumulative loss when using  $k$ instead of $n$ signals is at most $(n-k)/n$ and this is tight. Put differently, we show matching lower and upper bounds of $k/n$ on the fraction of the sender utility that can be obtained when using $k$ instead of $n$ signals. Up to small constant factors, similar bounds hold for independent instances.

More generally, extensions of algorithmic persuasion to multiple receivers have been studied by Babichenko and Barman~\cite{BabichenkoB17} and Arieli and Babichenko~\cite{ArieliB19}, who focus on private signals, as well as Dughmi and Xu~\cite{DughmiX17}, who contrast private and public signals.
Bhaskar et al.~\cite{BhaskarCKS16} and Rubinstein~\cite{Rubinstein17} study scenarios in which the receivers are players in games, proving various hardness results. 
Xu~\cite{Xu20} gives efficient approximation algorithms for some sub-classes of these scenarios.
Dughmi et al.~\cite{DughmiNPW19} employ Lagrangian duality to characterize near-optimal persuasion schemes, and study a further extension that includes payments.
Finally, to complement the multiple-receiver setting, multiple-sender settings have been studied in \cite{AuK20,GentzkowK17b,GentzkowK17a,LiN18,GradwohlHHS20}. 

A different approach was taken by Hahn et al.~\cite{HahnHS20IJCAI,HahnHS20}, who design approximation algorithms for online versions of the single-sender, single-receiver setting. In their models, the state of nature is revealed sequentially to $\sender$, $\sender$ sends a signal in each round to $\receiver$, and $\receiver$ then makes a binary decision.
Also somewhat related is the paper of Le Treust and Tomala~\cite{LeTreustT19}, who study a repeated setting with limited communication through a noisy channel. 


\section{Model}
\label{sec:model}
\paragraph{Signaling with Limited Messages}
There are two agents, a sender $\sender$ and a receiver $\receiver$. The receiver can take one of $n$ \emph{actions}. We denote the set of actions by $[n] = \{1,\dots,n\}$. Each action $i \in [n]$ has a \emph{type} $\state_i$ from a known type set $\Theta_i$. We assume throughout that all type sets are finite. The state of nature $\vecState = (\state_1,\dots,\state_n)$ is drawn according to a commonly known distribution over $\Theta$, where $\Theta \subseteq \Theta_1 \times \Theta_2 \times \dots \times \Theta_n$. We denote the probability of drawing the state $\vecState$ by $q_{\vecState}$.

Action $i$'s type $\state_i$ is associated with a value-pair $(\rv(\state_i),\sv(\state_i))$, where $\rv(\state_i)$ is the value for $\receiver$ and $\sv(\state_i)$ is the value for $\sender$ if action $i$ is taken by $\receiver$. Both agents want to maximize their respective expected utility from the action taken. While the distribution over states of nature is common knowledge, the realized state $\vecState$ is only observed by $\sender$. After observing $\vecState$, $\sender$ sends some abstract signal $\sigma \in \Sigma$ to $\receiver$.

We assume that $\sender$ has commitment power, i.e., $\sender$ commits in advance to a \emph{signaling scheme} $\phi$. It maps the observed state of nature $\vecState$ to a signal $\sigma$. More formally, $\phi(\state,\sigma)$ denotes the probability that in state $\state$ the scheme sends signal $\sigma$. $\phi$ is revealed to $\receiver$ before $\vecState$ is realized. The game we study proceeds as follows: (1) Both players know the prior distribution $q$. (2) $\sender$ commits to a signaling scheme $\phi$ and reveals it to $\receiver$. (3) The state of nature $\vecState$ is realized and is revealed to $\sender$. (4) $\sender$ draws signal $\sigma$ according to the distribution $\phi(\state,\cdot)$ and sends $\sigma$ to $\receiver$. (5) $\receiver$ chooses an action $i \in [n]$, and utilities are realized.

In the standard case of Bayesian persuasion with $|\Sigma| = k \ge n$, the sender can use signals to directly recommend every possible action to the receiver. In this paper, we are interested in $k < n$ when $\sender$ might not be able to directly recommend every single action to $\receiver$. Since the case of a single signal and $k=1$ is trivial, we assume $k \ge 2$ throughout.

We denote the expected utility for $\mathcal{X} \in \{\sender,\receiver\}$ by $u_{\mathcal{X}}(\phi)$ when $\sender$ uses scheme $\phi$ and $\receiver$ best responds to $\phi$ by picking, for every signal $\sigma$, an action with optimal expected utility conditioned on observing $\sigma$. Given $\sigma$, if $\receiver$ has several optimal actions, we assume $\receiver$ breaks ties in favor of the sender\footnote{This is a standard assumption in bilevel optimization problems. It is mainly used to avoid technicalities such as tiny perturbations to break ties.}. If within the set of actions with best utility for $\receiver$ there are several that have best utility for $\sender$, we assume w.l.o.g.\ that $\receiver$ chooses one of them via any fixed tie-breaking rule.

We will be interested in \emph{direct} and \emph{persuasive} schemes. In a direct scheme, $\sender$ uses each signal to recommend a single specific action. In a persuasive scheme, the receiver has no incentive to deviate from the recommended action. When considering persuasiveness, a useful quantity is the best expected utility of any fixed action for $\receiver$, which we denote by
%
$\rv_E = \max_{i \in [n]} \sum_{\vecState} q_{\vecState} \cdot \rv(\state_i).$
%

\vspace{-0.3cm}
\paragraph{Symmetric Instances} In a symmetric instance, any two states of nature that are a permutation of one another occur with the same probability.
Formally, in a symmetric instance, $q_{\vecState} = q_{\vecState'}$ whenever $\vecState'$ is any permutation of $\vecState$. In particular, due to symmetry, $\rv_E = \sum_{\vecState} q_{\vecState} \cdot \rv(\state_i)$ for every $i \in [n]$. 

Any symmetric distribution with finite type sets can be represented rather explicitly by a set of vectors, each having $n$ (not necessarily distinct) types, and a probability distribution over the vectors. A state of nature $\vecState$ is generated by drawing one of the vectors according to the distribution and then permuting the chosen vector uniformly at random. We denote by $d$ the number of vectors in the representation and call this a \emph{$d$-random-order scenario}. For $d=1$, we obtain the random-order scenario from the introduction.

However, there are also interesting symmetric distributions with a much more compact representation. For the IID scenario, the natural representation is only a type distribution for a single action from which we draw $n$ times to generate the state of nature. In the vector-based $d$-random-order representation, $d$ could be exponential in the number of types for a single action. Hence, we also study a more compactly represented \emph{prophet-secretary} scenario: Here we have $n$ (not necessarily distinct) distributions over types. The state of nature $\vecState$ is generated by an independent random draw from each of the $n$ distributions and a subsequent uniform random permutation of the $n$ types. The name stems from the literature on online algorithms. The prophet-secretary scenario strictly generalizes both IID and random-order scenarios. 

For simplicity, we will assume throughout that all types are indeed distinct. Note that this assumption will be without loss of generality, since we allow distinct types to be associated with the same pair of utility values for $\sender$ and $\receiver$.

\vspace{-0.3cm}
\paragraph{Independent Instances} In an independent instance, every action $i \in [n]$ has a type space $\Theta_i$. For simplicity we assume that the sets $\Theta_i$ are distinct, where we note that distinct types can have the same utility pairs. For each action $i \in [n]$ we have a distribution over types. We denote the probability of type $\theta_i \in \Theta_i$ by $q_{i,\theta_i}$. The state of nature $\vecState$ is generated by an independent draw from each of the $n$ distributions. 


\vspace{-0.3cm}
\paragraph{Direct and Persuasive}
We assume the sender has only $2 \le k \le n$ possible signals. Every instance with $k$ signals has an optimal direct and persuasive scheme. For symmetric instances we can assume these are the first $k$ actions. The proof is a simple revelation-principle-style argument and given in Appendix~\ref{app:directPersuasive}.
\begin{lemma}
    \label{lem:directPersuasive}
    There exists an optimal scheme with $k$ signals that is direct and persuasive and uses the signals to recommend $k$ distinct actions. In symmetric instances, there is an optimal direct and persuasive scheme in which $\sender$ recommends the actions from $[k]$.
\end{lemma}

\section{Symmetric Instances}

\subsection{Characterization of Optimal Schemes}
\label{sec:symmetric}
In this section, we derive a characterization of an optimal scheme in symmetric instances. Due to Lemma~\ref{lem:directPersuasive} we consider a direct and persuasive scheme that recommends actions from the set $[k]$. Suppose we are given a realization $\vecState$ of the state of nature. We interpret the action types as points in the two-dimensional plane. Type $\state_i$ corresponds to point $(\rv(\state_i), \sv(\state_i))$ . We use $C$ to denote the realized set of action types of the first $k$ actions.

Given any direct and persuasive scheme $\varphi$, consider the event that the state of nature gives rise to a set $C$ of types for the first $k$ actions. We denote the probability of this event by
%
   $q_C = \Pr\left[ \bigcup_{i \in [k]} \{\state_i\} = C \right].$
%
Conditioned on the set $C$ of types of the first $k$ actions, consider the point composed of the expected utilities for $\receiver$ and $\sender$, i.e., the point $(\Ex{u_\receiver(\varphi) \mid C}, \Ex{u_\sender(\varphi) \mid C})$. Graphically, this point lies inside the convex hull of the points of $C$. We term this the \emph{recommendation point for $C$} of $\varphi$.

More generally, let us define a \emph{point collection}. A point collection $\calP$ contains for each set $C$ of action types for the first $k$ actions a point $p(C) = (p_\receiver(C), p_\sender(C))$ inside the convex hull of $C$. We define the utilities of $\sender$ and $\receiver$ for $\calP$ by
\[
    u_\sender(\calP) = \sum_{C} q_C \cdot p_\sender(C) \qquad \text{and} \qquad u_\receiver(\calP) = \sum_{C} q_C \cdot p_\receiver(C)\enspace.
\]
Observe that the recommendation points of a direct and persuasive signaling scheme are a point collection, and the utility of the collection equals the utility of the scheme, for both $\sender$ and $\receiver$. However, in general a point collection might not correspond to a persuasive signaling scheme.

Our interest lies in point collections where, for every subset $C$, the point lies on the corresponding Pareto frontier of $C$. Graphically speaking, the Pareto frontier of $C$ can be assumed to start from a type with largest sender utility with a horizontal line (possibly of length 0) with slope 0 and end at a type with largest receiver utility with a vertical line (again, possibly of length 0) with slope $-\infty$. Hence, for every slope $s \in [0, -\infty]$, there is a point on the Pareto frontier such that a line with slope $s$ lies tangent to the Pareto frontier at this point. We say that a type or a point \emph{corresponds to a slope} $s$ if a line with slope $s$ lies tangent to the Pareto frontier in the point.

We concentrate on point collections that satisfy the following slope condition.
\begin{definition}
  \label{defi:sOptimal}
  For $s \le 0$, a point collection $\calP$ is \emph{$s$-Pareto} if (1) for every subset $C$, $p(C)$ is on the Pareto frontier of $C$ and corresponds to slope $s$ and (2) $u_\receiver(\calP) \ge \rv_E$.
\end{definition}
Our first main result is a characterization of an optimal scheme via an $s$-Pareto point collection.
\begin{theorem}
	\label{thm:sPareto}
	For every symmetric instance, there is an optimal scheme whose recommendation points are a sender-optimal $s$-Pareto point collection, over all $s \le 0$.
\end{theorem}
We prove the theorem using the following three lemmas. First, we show that for every persuasive scheme $\varphi$, there is an $s$-Pareto point collection $\calP$ with $u_\sender(\calP) \ge u_\sender(\varphi)$.

\begin{lemma}
    \label{lem:moveToPareto}
    For every direct and persuasive scheme $\varphi$, there is an $s$-Pareto point collection $\calP$ with $u_\sender(\calP) \ge u_\sender(\varphi)$.
\end{lemma}
\begin{proof}
    Consider an arbitrary persuasive scheme $\varphi$ that uses signals corresponding to the first $k$ actions. Let $\calP(\varphi)$ be the point collection of recommendation points of $\varphi$. Since $\varphi$ is persuasive, the collection $\calP(\varphi)$ satisfies the second condition of $s$-Pareto. Now we adjust $\calP(\varphi)$ in two steps to show the lemma.

    First, move every recommendation point up vertically to the Pareto frontier. This improves the sender utility of the point collection but keeps the receiver utility the same. Hence, the resulting point collection $\calP$ has all points on the Pareto frontiers, continues to satisfy $u_\receiver(\calP) \ge \rv_E$, and $u_\sender(\calP') \ge u_\sender(\varphi)$.

    Second, suppose there are different subsets $C_1 \neq C_2$ and there is no common slope that points $p(C_1)$ and $p(C_2)$ both correspond to. We use the short notation $p_1 = (\rv_1,\sv_1) = (p_\receiver(C_1), p_\sender(C_1))$ and $p_2 = (\rv_2,\sv_2) = (p_\receiver(C_1), p_\sender(C_2))$, respectively. In particular, suppose $p_1$ corresponds to slope $s_1$ and $p_2$ to slope $s_2 < s_1$. As the slopes are non-positive, $s_2$ is ``steeper'' than $s_1$.

    We construct a new point collection $\calP_1$. For any subset $C \neq C_1, C_2$ of types of the first $k$ actions, we keep $p(C)$. For sets $C_1$ and $C_2$ we adjust the points -- we set
    \[
      p_1' = (\rv_1 + \delta q_{C_2}, \, \sv_1 + \delta q_{C_2} s_1) \hspace{1cm} \text{and} \hspace{1cm} p_2' = (\rv_2 - \delta q_{C_1}, \, \sv_2 - \delta q_{C_1} s_2)
    \]
    by some sufficiently small $\delta > 0$. Intuitively, we move $p_1$ to the ``right'' for the set $C_1$ and to the ``left'' for $C_2$ -- thereby shifting the points along the segments on their respective Pareto frontiers.
    This implies that the sender utility of the point collection grows to

    \begin{align*}
        u_{\sender}(\calP_1) &= \sum_{C \neq C_1,C_2} q_C \cdot p_\sender(C) + q_{C_1} \cdot (\sv_1 + \delta q_{C_2} s_1) + q_{C_2}  \cdot (\sv_2 - \delta q_{C_1}s_2)\\
                           &= u_\sender(\calP) + q_{C_1} \delta q_{C_2} \cdot (s_1 - s_2) \quad > \quad u_\sender(\calP) \quad  \ge \quad u_\sender(\varphi)\enspace,
    \end{align*}
    since $0 \ge s_1 > s_2$. For the receiver utility
    \begin{align*}
        u_{\receiver}(\calP_1) &= \sum_{C \neq C_1,C_2} q_C \cdot p_\receiver(C) + q_{C_1}\cdot (\rv_1 + \delta q_{C_2}) + q_{C_2} \cdot (\rv_2 - \delta q_{C_1})\\
                           &= u_\receiver(\calP) + q_{C_1} \delta q_{C_2} - q_{C_2} \delta q_{C_1} \quad = \quad u_\receiver(\calP) \quad \ge \quad  \rv_E\enspace.
    \end{align*}
    Hence, $\calP_1$ satisfies the second property of $s$-Pareto, while improving the utility for the sender.

	$\delta$ is chosen such that $p_1'$ and $p_2'$ both stay on the line segments of slopes $s_1$ and $s_2$, respectively. Now repeated application of this modification yields collections $\calP_2, \calP_3,\ldots$ until finally points $p_1$ and $p_2$ correspond to at least one common slope: Whenever an endpoint of a line segment is reached, if this endpoint does not correspond to a slope of the other point, the process can be continued. Moreover, we can apply this modification repeatedly as long as there are two size-$k$-sets $C_1$, $C_2$ of types with points that have no common slope. Eventually, we reach an $s$-Pareto point collection $\calP$ with $u_\sender(\calP) \ge u_\sender(\varphi)$.
\end{proof}

Consider any $s$-Pareto point collection $\calP$. We define a direct scheme $\phi^*$ as follows: Given a set $C$ of types in the first $k$ actions and the point $p(C)$, $\phi^*$ recommends one of the (at most) two actions that compose the corresponding line segment of $p(C)$ on the Pareto frontier. The actions are chosen independently of their actual number within the first $k$ actions. By setting appropriate probabilities, the point $p(C)$ corresponds to the (conditioned on the given set $C$) expected utilities of $\phi^*$ for $\sender$ and $\receiver$. This directly implies that $u_\sender(\phi^*) = u_\sender(\calP)$ and $u_\receiver(\phi^*) = u_\receiver(\calP)$.

Due to symmetry of the instance and a choice of action independent of its number within the first $k$ actions, the scheme $\phi^*$ is symmetric. A \emph{symmetric scheme} $\varphi$ (see also~\cite{DughmiX16}) is direct and recommends with each signal a distinct action in $[k]$. The conditional distribution over types (resulting from the prior and $\varphi$) is the same for each recommended action. 
The conditional distribution over types is the same for each non-recommended action in $[k]$ and the same for each non-recommended action in $[n] \setminus [k]$, no matter which (other) action is recommended. 
Thus, a symmetric scheme gives rise to three distributions over types: a distribution $\dist_{yes}$ for any recommended action, a distribution $\dist_{no}$ for any non-recommended action in $[k]$, and a distribution $\dist_{never}$ for any non-recommended action in $[n] \setminus [k]$. For symmetric schemes, we show that persuasiveness is equivalent to the following simple constraint.

\begin{lemma}\label{lem:symPersuasive}
	In symmetric instances, a symmetric scheme $\varphi$ is persuasive if and only if $u_\receiver(\phi) \ge \rv_E$.
\end{lemma}
\begin{proof}
    Clearly, if a scheme $\varphi$ guarantees strictly less utility than $\rv_E$ to $\receiver$, then $\receiver$ could profit by deviating to, say, action 1 throughout. Hence, $u_\receiver(\phi) \ge \rv_E$ is necessary for every persuasive scheme $\phi$.

	Consider a symmetric scheme and the three resulting type distributions $\dist_{yes}$, $\dist_{no}$ and $\dist_{never}$. We denote by $\rv_{yes}$, $\rv_{no}$ and $\rv_{never}$ the expectations of the utility of $\receiver$ for the respective distributions. The previous lemma implies that if $\varphi$ is persuasive, then $\rv_{yes} \ge \rv_E$. Now, for the reverse direction, assume that $\rv_{yes} \ge \rv_E$. Clearly, since instance and scheme are symmetric, it holds that $\rv_{never} = \rv_E$. Again, due to symmetry, every action $i \in [k]$ gets recommended with probability $1/k$. Hence,
	$
	\frac{1}{k} \cdot \rv_{yes} + \frac{k-1}{k} \cdot \rv_{no} = \rv_E,
	$
	and $\rv_{yes} \ge \rv_E$ implies $\rv_{no} \le \rv_E$. It is not profitable for $\receiver$ to deviate from the recommended action. Hence, if $\rv_{yes} \ge \rv_E$, then $\varphi$ is persuasive.
\end{proof}

The symmetric scheme $\varphi^*$ based on an $s$-Pareto point collection satisfies the constraint in Lemma~\ref{lem:symPersuasive} by definition. As such, we obtain the following result, which finishes the proof of Theorem~\ref{thm:sPareto}.

\begin{lemma}
    \label{lem:persuasive}
    For every $s$-Pareto point collection $\calP$, there is a symmetric, direct, and persuasive signaling scheme $\phi^*$ with $u_\sender(\phi^*) = u_\sender(\calP)$.
\end{lemma}



\subsection{Efficient Computation of Optimal Schemes}
\label{sec:symmetricCompute}

\label{sec:slopeAlgo}

The Slope-Algorithm (Algorithm~\ref{algo:slope}) systematically enumerates a set $S$ containing all meaningful candidate slopes $s$ for an $s$-Pareto point collection. For every pair of types $a,b$ the algorithm determines the probability (denoted by $p_{ab}$) that their line segment (denoted by $\overline{ab}$) is contained in the Pareto frontier of the set $C$ of realizations of the first $k$ actions. For every pair with $s > 0$, one type Pareto dominates the other and the pair can be discarded. Similarly, if $p_{ab} = 0$, the pair can be discarded. The critical step in the first part of the algorithm is the computation of $p_{ab}$ in line~\ref{line:prob1}. For now, we assume that the algorithm has oracle access to these quantities via a \emph{probability oracle}. We will discuss below how to implement the probability oracle in polynomial time.

\begin{algorithm}[t!]
	\KwIn{Symmetric instance with set $\Theta = \Theta_1 = \ldots = \Theta_n$ of action types and distribution $q$}
    $S \leftarrow \emptyset$, $L \leftarrow \emptyset$\\
	\For{every pair of types $a,b \in \Theta$, $a \neq b$}{
        Let $s$ be the slope of $\overline{ab}$ and set $p_{ab} \leftarrow 0$. \\
        \textbf{if} $s \le 0$ \textbf{then} determine prob.\ $p_{ab}$ that $\overline{ab}$ is on the Pareto frontier of types of actions in $[k]$ \label{line:prob1}\\
        \textbf{if} $p_{ab} > 0$ \textbf{then} $S \leftarrow S \cup \{s\}$\\
    }
    Sort the slopes of $S$: $s_1 < s_2 < \ldots < s_{\ell}$\\
    Pick $\ell+1$ auxiliary slopes: $t_1 < s_1 < t_2 < s_2 < \ldots < s_{\ell} < t_{\ell+1}$ \label{line:auxSlope}\\
    $S \leftarrow S \cup \{t_1,\ldots,t_{\ell+1}\}$\\
    \For{every slope $s \in S$}{
       \For{every type $c \in \Theta$}{
         Determine probability $p_c^{(s)}$ that $c$ is the unique point corresponding to $s$ on the Pareto frontier of types of actions in $[k]$\label{line:prob2}
        }
      Solve the following LP to determine an $s$-Pareto point collection:
    \begin{equation}
        \label{eq:slopeLP}
        \begin{array}{lrcll}
        \mbox{Max.} & \multicolumn{2}{l}{\displaystyle\sum_{\overset{c,d \in \Theta, c \neq d}{\overline{cd} \mbox{\scriptsize{ has slope }} s}} p_{cd} \cdot \left(\alpha^{(s)}_{cd} \sv_c + (1-\alpha^{(s)}_{cd}) \sv_d \right)  + \sum_{c \in \Theta}{p_c^{(s)} \sv_c }  }\\
        \mbox{s.t.} & \displaystyle\sum_{\overset{c,d \in \Theta, c\neq d}{\overline{cd} \mbox{\scriptsize{ has slope }} s}} p_{cd} \cdot\left( \alpha^{(s)}_{cd} \rv_c + (1-\alpha^{(s)}_{cd}) \rv_d \right) + \sum_{c \in \Theta} p^{(s)}_c \rv_c & \ge & \rv_E \\
        & \alpha^{(s)}_{cd}    \in [0,1] \mbox{ for all } c,d \in \Theta
        \end{array}
    \end{equation}
    \textbf{if} LP~\eqref{eq:slopeLP} has feasible optimal solution $\vecAlpha^{(s)}$ \textbf{then} $L \leftarrow \{(\vecAlpha^{(s)}, s)\}.$
    }
    \Return{best point collection in $L$ with corresponding slope}
    \caption{Slope-Algorithm}
	\label{algo:slope}
\end{algorithm}

At the end of the first for-loop, the algorithm has collected in $S$ all meaningful slopes of non-empty segments that can appear on the Pareto frontier of the types of the first $k$ actions. In addition to these slopes, every Pareto-frontier can be assumed to contain all slopes from $[0, -\infty)$. An optimal scheme might not necessarily correspond to a slope of any non-empty segment attained in the first for-loop.
If it does not, it must correspond to some slope $t$ with $s_i < t < s_{i+1}$. Note that all slopes $t \in (s_{i},s_{i+1})$ correspond to the same point on the Pareto frontier.
Hence, $t_i$ in line 7 can be chosen arbitrarily.
%

Now even if a slope $s$ is attained by some segment $\overline{ab}$, it might be that for some other subset of types $C$, slope $s$ only corresponds to a single point on the Pareto frontier of $C$. As such, the algorithm also determines in line~\ref{line:prob2} for every $s \in S$ the probability that a single type $c \in \Theta$ corresponds to $s$ on the Pareto frontier of $C$. This is the critical step in the second part of the algorithm. Again, we assume that the algorithm has oracle access to these quantities via a probability oracle. We will discuss in the next section how to implement the probability oracle in polynomial time.

Finally, after having computed all probabilities the algorithm solves LP~\eqref{eq:slopeLP}. For the LP we assume that $s$ is the common slope of the point collection. Clearly, for all subsets $C$ where a single point $c$ corresponds to slope $s$, the choice is trivial. For all subsets $C$, in which some line segment $\overline{cd}$ with slope $s$ is on the Pareto frontier, there is a choice to pick a point from that segment. This choice is represented by the variable $\alpha_{cd}^{(s)} \in [0,1]$. The LP optimizes point locations to maximize the expected utility for $\sender$ (in the objective function) and to guarantee at least the average utility of $\rv_E$ for $\receiver$. For a given slope $s$, the LP might be infeasible. However, by enumerating all relevant common slopes, the algorithm sees at least one feasible solution. It returns the best feasible LP solution along with the slope $s^*$.

Note that the output of the algorithm is sufficient for $\sender$ to implement an optimal persuasive scheme. $\sender$ looks at the set $C$ of the types of the first $k$ actions, computes the Pareto frontier, and looks for slope $s^*$. If $s^*$ is realized by a segment $\overline{ab}$, $\sender$ recommends the action with type $a$ with probability $\alpha^{(s^*)}_{ab}$ and the action with type $b$ with probability $1-\alpha^{(s^*)}_{ab}$. If it is realized through a single type $c$, $\sender$ recommends this action with probability 1.

\begin{proposition}
   \label{prop:Slope}
   Given an efficient algorithm to compute the probability oracle, the Slope-Algorithm computes an optimal direct and persuasive scheme for symmetric instances in polynomial time.
\end{proposition}

\begin{proof}
   Correctness follows from the characterization in the last section and the observations above. We denote the maximal running time of the probability oracle by $T_o$ and the maximal time needed to solve LP~\eqref{eq:slopeLP} by $T_{LP}$. Let $m = |\Theta|$ denote the finite number of types. Then finding the slopes can be done in time $O(m^2\cdot T_o)$. Sorting the slopes needs time $O(m \log m)$. For the second for-loop, we iterate through $O(m^2)$ slopes. For each slope, we need at most $m$ calls to the probability oracle and solve one LP of polynomial size. Overall, the running time is $O(m^3 \cdot T_o + m^2 \cdot T_{LP})$.
\end{proof}

Using geometric properties of the utility pairs in prophet-secretary and $d$-random-order scenarios, we show how to design polynomial-time probability oracles in these scenarios. For full proofs see Appendix~\ref{app:efficientProbability}.

\begin{theorem}\label{thm:symmetricPolyTime}
	An optimal signaling scheme with $k$ signals can be computed in polynomial time for the prophet-secretary and the $d$-random-order scenarios.
\end{theorem}

\section{Independent Instances}
\label{sec:independent}


In this section, we move away from symmetric instances and concentrate on the case of independent actions.
For such instances, computing the expected utility for $\sender$ is \classSharpP-hard, even in the standard case with $n$ actions and $n$ signals~\cite{DughmiX16}. We discuss how to obtain a persuasive scheme for $k$ signals that guarantees a constant-factor approximation to the optimal sender utility for $k$ signals.

We first identify an action with the highest a-priori utility $\rv_E$ for $\receiver$. If there are multiple such actions, pick one that maximizes the expected utility for $\sender$. If there are several of these, pick an arbitrary one from these. We re-number the actions such that this is action $n$. Our signaling schemes use $k$ signals to recommend a set $S \cup \{n\}$ of $k$ actions. The signal for action $n$ plays the role of a dummy signal (c.f.~\cite{DughmiKQ16}).

Our algorithm applies in independent instances, in which there is an optimal scheme $\varphi^*$ such that $\calR$ obtains a conditional expectation of at least $\rv_E$ for every signal. We term this condition \emph{$\rv_E$-optimality}. For example, $\rv_E$-optimality is fulfilled when there is an action that has deterministic utility of $\rv_E$ for $\calR$ (but possibly randomized utility for $\calS$). Then $\calR$ can always secure a value of $\rv_E$ by choosing this action. As such, to be persuasive, $\varphi^*$ must guarantee at least a conditional expected utility of $\rv_E$ for every signal.

Our signaling schemes consist of two steps: (a) choose a suitable set $S$ of $k-1$ actions, and (b) given any set $S \cup \{n\}$ of $k$ actions, compute a signaling scheme that recommends one of these actions.
We give two variants that follow this approach. First, in Section~\ref{sec:constantFactor} we consider the Independent Scheme $\varphi_{IS}$ based on a greedy algorithm for step (a). The approximation guarantee is given in the subsequent theorem. It is $3/8 = 0.375$ for $k=2$. For $k \to \infty$ it approaches $(1-1/e)^2 \approx 0.3996$.
\begin{theorem}
	\label{thm:independentGreedy}
	The Independent Scheme $\varphi_{IS}$ is a direct and persuasive scheme for $\rv_E$-optimal independent instances with $k$ signals. It can be implemented in time polynomial in the input size. For every $k \ge 2$,
	\[
	u_\sender(\varphi_{IS}) \ge \left(1-\left(1-\frac{1}{k}\right)^k\right) \cdot \left(1-\left(1-\frac{1}{k}\right)^{k-1}\right) \cdot u_\sender(\varphi^*)\enspace.
	\]
\end{theorem}
Subsequently, in Section~\ref{sec:FPTAS}, we describe an improved procedure to compute a good set $S$ in step (a). This improves the approximation ratio considerably for larger values of $k$. The ratio is at least $0.375-\varepsilon$ for $k=2$. For $k \to \infty$, it is at least $1 - 1/e - \varepsilon$.
\begin{theorem}
	\label{thm:independentFPTAS}
	The Improved Independent Scheme $\varphi_{IIS}$ is a direct and persuasive scheme for $\rv_E$-optimal independent instances with $k$ signals. It can be implemented in time polynomial in the input size. For every $k \ge 2$ and every constant $\varepsilon > 0$
	\[
	u_\sender(\varphi_{IIS}) \ge  \left(1-\left(1-\frac{1}{k}\right)^{k}\right) \cdot (1-\varepsilon) \cdot \left(1-\frac{1}{k}\right) \cdot u_\sender(\varphi^*)\enspace.
	\]
\end{theorem}
We observe below that for large values of $k$, this is essentially a tight guarantee for our approach. A further improvement of the approximation ratio requires significantly different techniques.

\subsection{Constant-Factor Approximation}
\label{sec:constantFactor}

In this section, we describe the Independent Scheme and prove Theorem~\ref{thm:independentGreedy}. For each type set $\Theta_i$, we w.l.o.g.\ include a sufficient number of dummy types $\state_i$ with $q_{i,\state_i} = 0$ and assume that $|\Theta_i| = |\Theta_j| = m$, for all $i,j \in [n]$. We use $[m]$ to enumerate the possible types of each action $i$. Now for any subset $S \subseteq [n-1]$ of the first $n-1$ actions, consider a set function $f : 2^{[n-1]} \to \mathbb{R}$ defined by
\begin{equation}
\label{eqn:f}
f(S) = \max \left\{ \sum_{i \in S \cup \{n\}} g_i(z_i) \growingmid \sum_{i \in S \cup \{n\}} z_i \le 1 \text{ and } z_i \ge 0 \text{ for all } i \in S \cup \{n\}\right\}
\end{equation}
where
\begin{equation}\label{eqn:LP-submodular2}
\begin{array}{llrcll}
\renewcommand{\arraystretch}{0.5}
g_i(z) = &\mbox{Max.} & \multicolumn{2}{l}{\D \sum_{j=1}^m x_{ij} \sv_{ij} }\\
& \mbox{s.t.} & \D \sum_{j=1}^m x_{ij} & \le & z \\
& & \D \sum_{j=1}^m x_{ij} \rv_{ij} & \ge & \rho_{E}\cdot \D \sum_{j=1}^m x_{ij} & \\
& & x_{ij} & \in & [0,q_{ij}] & \mbox{for all } j \in [m] 
\end{array}
\end{equation}
For an intuition, we interpret $z_i$ as an overall probability of a signal for action $i$. Then $g_i(z_i)$ maximizes the expected utility for the sender conditioned on a probability mass of $z_i$ on action $i$. In LP~\eqref{eqn:LP-submodular2}, $x_{ij}$ describes the portion of the probability mass on type $j$ of action $i$. The first constraint of LP~\eqref{eqn:LP-submodular2} limits the total mass of action $i$ to at most $z$. The second constraint ensures that the conditional expected utility of $\vecx$ for $\receiver$ is at least $\rv_E$. Finally, the last constraint states that the probability of a signal for type $j$ is at most the probability that type $j$ is realized.

Consider any direct and persuasive scheme $\varphi_{S \cup \{n\}}$ that uses $|S|+1$ signals to recommend the actions $S \cup \{n\}$. Suppose $x_{ij}$ is the ex-post probability to recommend action $i$ with type $j$ in $\varphi_{S \cup \{n\}}$. Clearly, the constraints in~\eqref{eqn:f} and~\eqref{eqn:LP-submodular2} do not fully capture the constraints on $x_{ij}$. However, all constraints are necessary. In particular, setting $x_{ij}$ to the ex-post probability of recommending action $i$ with type $j$ in the optimal scheme $\varphi^*_{S \cup \{n\}}$ gives a feasible solution for every LP~\eqref{eqn:LP-submodular2}, and $z_i = \sum_{j=1}^m x_{ij}$ is feasible for~\eqref{eqn:f} (c.f.~\cite[Lemma 1]{HahnHS20IJCAI}). Hence, for any given subset $S \cup \{n\}$ of recommended actions, $f(S)$ is an upper bound on the optimal sender utility, i.e., $f(S) \ge u_\sender(\varphi^*_{S \cup \{n\}})$.


Now, consider the Independent Scheme $\varphi_{IS}$. It consists of two steps: (a) choose a suitable set $S$ of $k-1$ actions, and (b) given any set $S \cup \{n\}$ of $k$ actions, compute a signaling scheme that recommends one of these actions. Step (a) is done in ActionsGreedy (Algorithm~\ref{algo:greedy}), step (b) in ComputeSignal (Algorithm~\ref{algo:signal}).
\begin{algorithm}[t]
   \DontPrintSemicolon
	\KwIn{Type sets $\Theta_1, \ldots, \Theta_n$ and distributions $q_1,\ldots,q_n$, s.t.\ $\sum_j q_{n,j} \rv_{nj} = \rv_E$ and $\sum_j q_{n,j}\sv_{nj} = \max_{i \in [n] \, : \, \sum_j q_{i,j}\rv_{ij} = \rv_E} \sum_j q_{i,j}\sv_{ij}$, parameter $2 \le k \le n$}
    $S \leftarrow \emptyset$\;
    \bfFor\ $\ell=1,\ldots,k-1$ \bfDo: Let $i$ be an action maximizing $f(S \cup \{i\}) - f(S)$ and set $S \leftarrow S \cup \{i\}$\;
    \Return{$S$}
   \caption{\label{algo:greedy} ActionsGreedy}
\end{algorithm}
\begin{algorithm}[t]
   \DontPrintSemicolon
	\KwIn{Type sets $\Theta_1, \ldots, \Theta_n$ and distributions $q_1,\ldots,q_n$, s.t.\ $\sum_j q_{n,j} \rv_{nj} = \rv_E$ and $\sum_j q_{n,j}\sv_{nj} = \max_{i \in [n] \, : \, \sum_j q_{i,j}\rv_{ij} = \rv_E} \sum_j q_{i,j}\sv_{ij}$, parameter $2 \le k \le n$, set $S \subseteq [n-1]$ with $|S| = k-1$}

    For $i \in S \cup \{n\}$, let $z^*_i$ and $\vecx^*_i$ be the values of the optimal solution in $f(S)$.\;
    Order actions in $S \cup \{n\}$ such that $\frac{g_{i_1}(z^*_{i_1})}{z^*_{i_1}} \ge \ldots \ge \frac{g_{i_{k+1}}(z^*_{i_k})}{z^*_{i_k}}$, where we assume $\frac{0}{0} = 0$\;

    \For{$\ell = 1,\ldots,k$}{
      Observe type $j$ of action $i_\ell$. Flip independent coin with probability $x^*_{i_\ell,j} / q_{i_\ell,j}$ for heads.\;
      \bfIf\ coin comes up heads \bfThen\ \Return{signal for action $i_\ell$}
    }
    \Return{signal for action $n$} \label{line:finalSignal}

   \caption{\label{algo:signal} ComputeSignal}
\end{algorithm}

We start our analysis by bounding the approximation of $\varphi_{IS}$ in terms of optimal sender utility. Towards this end, we observe that ActionsGreedy implements the greedy algorithm for submodular maximization.
\begin{lemma}
   $f$ is non-negative, non-decreasing, and submodular.
\end{lemma}

\begin{proof}

$f$ is clearly non-negative and non-decreasing -- since every $g_j$ is non-negative, piece-wise linear, and concave. Hence, $f(S \cup \{j\})$ can only improve over $f(S)$. To see that $f$ is submodular, note that $f$ optimally distributes a unit of mass to a set of monotone, concave functions. Consider the common slope of the functions $g_i$ for $i \in S$ resulting from the optimal waterfilling assignment of $z^*_i$ in $f(S)$. When going from $S$ to $S \cup \{j\}$, the slope can only decrease. As a consequence, when adding more elements to $S$, the $z^*_i$ are non-increasing.

Consider $S \subseteq T$ and $j \not\in T$. Let $z^S_j$ be the optimal choice in $f(S \cup \{j\})$ and $z^T_j$ be the one in $f(T \cup \{j\})$. Note that $z^S_j \ge z^T_j$. Now assume that for $f'(S \cup \{j\})$, we only allow to assign at most $z^T_j$ to $g_j$. Then $f'(S \cup \{j\}) \le f(S \cup \{j\})$, since in the former a mass of $z^S_j - z^T_j$ yields a smaller growth in value due assignment to $i \neq j$ with a smaller slope. When shifting from $f(S)$ to $f'(S \cup \{j\})$ and from $f(T)$ to $f(T \cup \{j\})$, in both cases the increase at $j$ is $g_j(z'_j)$, and a mass of $z^T_j$ is removed from the remaining functions. This has a stronger effect in $S$, since the removal occurs at a higher slope. Overall, $f(T \cup \{j\}) - f(T) \; \le \; f'(S \cup \{j\}) - f(S) \; \le \; f(S \cup \{j\}) - f(S)$.
\end{proof}
By Lemma~\ref{lem:directPersuasive} we can assume that the optimal scheme $\varphi^*$ directly recommends a set $K$ of $k$ actions.
\begin{lemma}
  \label{lem:greedyApx}
  For every $k \ge 2$, ActionsGreedy computes a subset $S$ of $k-1$ actions such that
    \[  f(S) \ge \left(1 - \left(1-\frac 1k \right)^{k-1}\right) \cdot u_\sender(\varphi^*)\enspace.  \]
\end{lemma}
\begin{proof}
ActionsGreedy is a standard greedy algorithm for submodular maximization. Note that
\[
u_\sender(\varphi^*) \; \le \; u_\sender(\varphi^*_{K \cup \{n\}}) \; \le \; f(K) \; \le \; f(S_k^*)\enspace,
\]
where $S_k^* \in \arg \max \{ f(S) \mid S \subseteq [n-1], |S| = k \}$. The action $n$ is apriori receiver-optimal, and in our scheme below it will play the role of an outside option, a baseline or dummy signal (c.f.\ \cite{DughmiKQ16,HahnHS20IJCAI}). However, it is not necessarily part of the optimal subset $K$ of signals. As such, we overestimate the optimal value by $f(S_k^*)$, the best set of $k+1$ recommended actions, one of which must be action $n$.

A simple generalization of the standard analysis in~\cite{NemhauserWF78} (see, e.g.,~\cite[Theorem 1.5]{KrauseG14}) shows that for this case the greedy solution $S$ guarantees $f(S) \ge (1- (1-1/k)^{k-1}) \cdot f(S^*_k)$, and the lemma follows.
\end{proof}

Now consider the second step of $\varphi_{IS}$, i.e., the computation of a signal using ComputeSignal.

\begin{lemma}
  \label{lem:signalApx}
  For every $k \ge 2$, let $S \cup \{n\}$ be any set of $k$ actions. Given the set $S \cup \{n\}$ of actions, ComputeSignal computes a signaling scheme $\varphi$ such that
  \[ u_\sender(\varphi) \ge \left(1 - \left(1- \frac 1k \right)^k\right) \cdot f(S)\enspace.\]
\end{lemma}
\begin{proof}
Given the chosen set $S$ of actions, we consider these actions one-by-one in non-decreasing order of $g_i(z^*_i)/z^*_i$. ComputeSignal flips an independent coin for each action whether or not to recommend it. We perform several bounding steps to provide a lower bound on $u_\sender(\varphi)$. First, we assume that the final ``backup signal'' for action $n$ in the last line~\ref{line:finalSignal} has value 0 for $\sender$. We use $p_{\ell} = \prod_{\ell' = 1}^{\ell-1} (1-z^*_{\ell'})$ to denote the probability to arrive in iteration $\ell > 1$ in the for-loop. Conditioned on arriving in iteration $\ell$, the combined probability of action $i_{\ell}$ having state $j$ and issuing a recommendation is $q_{i_\ell,j} \cdot \frac{x^*_{i_\ell,j}}{q_{i_\ell,j}} = x^*_{i_\ell,j}$. Thus, conditioned on arriving in iteration $\ell$, the expected value for $\sender$ from this iteration is $\sum_{j=1}^m x^*_{i_\ell} \sv_{i_\ell,j} = g_{i_\ell}(z^*_{i_\ell})$. Overall,
\begin{equation}
\label{eqn:mediant}
\frac{u_\sender(\varphi)}{f(S)} \ge \frac{\D \sum_{\ell = 1}^k g_{i_\ell}(z^*_{i_\ell}) \cdot p_{\ell}}{\D \sum_{\ell=1}^k g_{i_\ell}(z^*_{i_\ell})} = \frac{\D \sum_{\ell = 1}^k u_{i_\ell} \cdot z^*_{i_\ell} \cdot p_{\ell}}{\D \sum_{\ell=1}^k u_{i_\ell} \cdot z^*_{i_\ell}}\enspace,
\end{equation}
where we use the notation $u_{i_\ell} = g_{i_\ell}(z^*_{i_\ell})/z^*_{i_\ell}$. Note that if $z = 0$, then $g_{i_\ell}(z) = 0$. More generally, if there is an action $i_\ell \in S$ with $g_{i_\ell}(z^*_{i_\ell}) = 0$, then we can drop it from consideration and consider the ratio with the $k-1$ remaining actions. Hence, we can assume that $u_{i_\ell} > 0$, for all $1 \le \ell \le k$. By scaling the terms, we obtain $u_{i_k} = 1$ without changing the ratio. Note that the last ratio in~\eqref{eqn:mediant} is a weighted mediant, where the terms $u_{i_\ell}$, $1 \le \ell \le k$, act as weights for the ratios
\begin{align*}
\frac{z^*_{i_1}}{z^*_{i_1}} > \frac{ z^*_{i_2} p_1}{z^*_{i_2}} > \ldots > \frac{z^*_{i_k} p_k}{z^*_{i_k}}\enspace.
\end{align*}
Repeated application of the generalized mediant inequality shows that when $u_{i_1} \ge \ldots \ge u_{i_k} = 1$, the ratio is minimized for $u_{i_1} = \ldots = u_{i_k} = 1$, i.e.,
\begin{align*}
\frac{u_\sender(\varphi)}{f(S)} \; &\ge \; \frac{\D \sum_{\ell = 1}^k u_{i_\ell} \cdot z^*_{i_\ell} \cdot p_{\ell}}{\D \sum_{\ell=1}^k u_{i_\ell} \cdot z^*_{i_\ell}}
 \; \ge \; \frac{\D \sum_{\ell = 1}^k z^*_{i_\ell} \cdot p_{\ell}}{\D \sum_{\ell = 1}^k z^*_{i_\ell}}  \; =  \; \sum_{\ell = 1}^k z^*_{i_\ell} \cdot p_{\ell}  \; =  \; 1 - \left(\sum_{i=1}^{k-1} z^*_{i_\ell}\right) \prod_{i=1}^{k-1} (1-z^*_{i_\ell})\\ &\ge  \; 1 - \frac{k-1}{k} \left(1 - \frac{1}{k}\right)^{k-1}  \; =  \; 1 - \left(1 - \frac{1}{k}\right)^k
\enspace.
\end{align*}
For the second line, observe that the last function in the first line is symmetric and convex in every variable $z^*_{i_\ell}$. As such, it has a global minimum at $z^*_{i_1}= \ldots = z^*_{i_k} = 1/k$.
\end{proof}

Combining the previous lemmas allows to bound the approximation ratio. We proceed to show persuasiveness of the scheme.

\begin{lemma}
  \label{lem:signalPersuasive}
  ComputeSignal returns a direct and persuasive signaling scheme for independent instances with $k$ signals.
\end{lemma}

\begin{proof}
	Note that ComputeSignal solves LP~\eqref{eqn:f} to optimality. Hence, due to the first constraint of LP~\eqref{eqn:LP-submodular2} we have $\sum_{j=1}^m x_{ij} \le z_i$ for every $i \in S \cup \{n\}$. We first argue that we can w.l.o.g.\ assume that this constraint holds with equality.

	Every LP~\eqref{eqn:LP-submodular2} is a parametric linear program. Increasing scalar $z$ increases the right-hand side of the first packing constraint. It is easy to see that $g_i(0) = 0$. Standard sensitivity analysis for parametric linear programs implies that $g_i(z)$ is non-decreasing, piece-wise linear, and concave. Hence, an optimal assignment $\vecz^*$ in~\eqref{eqn:f} results from a waterfilling approach, where we raise the entries $z^*_i$ until they sum up to 1, while keeping a common slope for all functions $g_i$ for $i \in S \cup \{n\}$ (w.l.o.g.\ we assume that a breakpoint between linear segments in $g_i$ represents all intermediate slopes). For every $i \in [n-1]$, there exists at most one breakpoint $\hat{z}_i \in [0,1]$ such that the slope of $g_i(z)$ is 0 for all $\hat{z}_i \le z \le 1$. If no such breakpoint exists, we can set $\hat{z}_i = 1$. W.l.o.g.\ we assume $0 \le z_i^* \le \hat{z}_i$ and $z_n^* \ge 0$ such that $\sum_{i=1}^n z_i^* = 1$. Observe that for every $z_n \in [0,1]$ we can assume the first constraint in LP~\eqref{eqn:LP-submodular2} holds with tightness without violating the second constraint with $\rv_E$. As a consequence, we can assume w.l.o.g.\ for every $i \in [n]$ that in the optimal solution $\vecz^*$ of $\eqref{eqn:f}$ the first constraint of every LP~\eqref{eqn:LP-submodular2} is satisfied with equality $\sum_{j=1}^m x^*_{ij} = z_i^*$.

	Using this insight, we prove persuasiveness. In particular, for every choice of the set $S$ of actions with $S \subseteq [n-1]$ with $|S| = k-1$, we show that ComputeSignal computes a direct and persuasive signal.

	For each action $i \in S \cup \{n\}$ ComputeSignal observes the type realization and uses the optimal solution $\vecx^*$ for LP~\eqref{eqn:LP-submodular2} to flip an independent coin that yields the recommendation for action $i$. First, condition on the event that the scheme returns the signal for action $i_{\ell} \in S$ in the last for-loop. We again use $p_\ell = \prod_{\ell'=1}^{\ell-1} (1-z^*_{i_{\ell'}})$ to denote the probability that the scheme arrives in iteration $\ell$. Due to independent coin flips in the for-loop, the probability that the signal is sent in iteration $\ell$ is $\sum_{j=1}^m q_{i_{\ell},j} \cdot x^*_{i_{\ell},j}/q_{i_{\ell},j} = z^*_{i_{\ell}}$, where we assume the equality $z^*_i = \sum_{j=1}^m x^*_{ij}$ as observed above. A signal for action $i_{\ell} \neq n$ yields a conditional expected utility for $\receiver$ of
  \[
    \frac{1}{p_{\ell} \cdot z^*_{i_{\ell}}} \cdot p_{\ell} \cdot \sum_{j=1}^m q_{i_{\ell},j} \cdot (x^*_{i_{\ell},j}/q_{i_{\ell},j}) \cdot \rv_{i_{\ell},j} = \frac{1}{z^*_{i_{\ell}}} \sum_{j=1}^m x^*_{i_{\ell},j}  \rv_{i_{\ell},j} \ge \rv_E\enspace,
  \]
  where the inequality follows from the second constraint in~\eqref{eqn:LP-submodular2}.

  Now suppose ComputeSignal signals action $n$. First, suppose the signal results from the last line of the scheme. Then all coins in other iterations $\ell' \neq \ell$ with $i_{\ell'} \neq n$ have not come up heads, which has probability $p_{-\ell} =  \prod_{\ell'\neq \ell} (1-z^*_{i_{\ell'}})$. In addition, the signal in iteration $\ell$ with $i_{\ell} = n$ must not be sent. $\receiver$ obtains an expected utility of
  \[
    p_{-\ell} \cdot \sum_{j=1}^m q_{n,j} \cdot \left(1 - x^*_{nj}/q_{n,j}\right) \cdot \rv_{nj} = p_{-\ell} \cdot \left(\rv_E - \sum_{j=1}^m x^*_{nj}\rv_{nj}\right)\enspace.
  \]
  Second, assume the signal results from iteration $\ell$ of the for-loop, then the expected utility is
  \[
    p_{\ell} \cdot \sum_{j=1}^m q_{i_{\ell},j} \cdot (x^*_{i_{\ell},j}/q_{i_{\ell},j}) \cdot \rv_{i_{\ell},j} = p_\ell \sum_{j=1}^m x^*_{i_{\ell},j} \rv_{i_{\ell},j}\enspace.
  \]
  A signal for action $n$ yields a conditional expected utility for $\receiver$ of
  \begin{align*}
    \frac{p_{\ell} \D \sum_{j=1}^m x^*_{nj} \rv_{nj}
    + p_{-\ell} \left(\rv_E - \sum_{j=1}^m x^*_{nj} \rv_{nj}\right)}{p_{\ell}\cdot z^*_n + p_{-\ell}\cdot(1-z^*_n)} &= \frac{p_{-\ell}\cdot \rv_E + (p_{\ell}-p_{-\ell}) \D \sum_{j=1}^m x^*_{nj} \rv_{nj}}{p_{-\ell} +(p_{\ell} - p_{-\ell}) \cdot z_n^*}\\ &\ge \frac{\rv_E \cdot (p_{-\ell} + (p_{\ell}-p_{-\ell}) \cdot z_n^*)}{p_{-\ell} + (p_{\ell} - p_{-\ell})\cdot z_n^*} = \rv_E\enspace,
  \end{align*}
  where the inequality follows from the equality $z^*_i = \sum_{j=1}^m x^*_{ij}$ and the second constraint in~\eqref{eqn:LP-submodular2}.

  Hence, for every recommended action, the expected value for $\receiver$ is at least $\rv_E$. Thus, deviating to any action $i \not\in S \cup \{n\}$ is not profitable for $\receiver$, since the type of action $i$ is independent of the signal, and every action a priori has expected value at most $\rv_E$ for $\receiver$.

  We condition on the case that ComputeSignal sends a signal for action $i_\ell \neq n$ in the for-loop. The expected value of action $i_{\ell'}$ with $\ell' > \ell$ is at most $\rv_E$, since the type of action $i_{\ell'}$ has not been observed. For $\ell' < \ell$, the scheme decided not to send a signal using an independent coin flip in iteration $\ell'$. The overall value of action $i_{\ell'}$ for $\receiver$ is most $\rv_E$, the value of a signal is at least $\rv_E$, so a non-signal for action $i_{\ell'}$ has value at most $\rv_E$ for $\receiver$. Similar arguments show that conditioned on a signal for action $n$, every other action has expected value at most $\rv_E$. This proves that the resulting scheme is persuasive.
\end{proof}

In terms of running time, GreedyActions solves~\eqref{eqn:f} an $O(nk)$ number of times. ComputeSignal solves~\eqref{eqn:f} only once, and then computes at most $k-1$ independent coin flips. Clearly, both algorithms can be implemented to run in time polynomial in the representation of the input. This concludes the proof of Theorem~\ref{thm:independentGreedy}.


\subsection{Improved Approximation and Tightness}
\label{sec:FPTAS}

In this section, we improve the approximation ratio of the scheme from the previous section. It is easy to see that Lemma~\ref{lem:signalApx} is tight -- there are cases\footnote{Consider a set $S \cup \{n\}$ consisting of $k$ IID actions. Every action $i \in S \cup \{n\}$ has two possible types $\Theta^i = \{\theta_1,\theta_0\}$, where $(\rv(\theta_1),\sv(\theta_1)) = (1,1)$, $q_{\theta_1} = 1/k$, and $(\rv(\theta_0),\sv(\theta_0)) = (0,0)$. Observe that $f(S) = 1$. The best persuasive scheme recommends an action with type $\theta_1$ whenever there is one, which happens only with probability $1-(1-1/k)^k$.} in which the sender utility of any persuasive scheme for action set $S \cup \{n\}$ can indeed recover at most a fraction of $1-(1-1/k)^k$ of $f(S)$.

Instead, we replace the standard greedy algorithm for submodular maximization by a more elaborate procedure to carefully choose a subset of actions. In this section, we describe an FPTAS to compute, for every given constant $\varepsilon > 0$, a set $S \subseteq [n-1]$ of $k-1$ actions such that $f(S) \ge (1-\varepsilon) \cdot f(S^*)$ for the set $S^* \subseteq [n-1]$ of $k-1$ actions that maximizes $f$.

Our approach in the algorithm described below is to use a discretized version $\hat{f}$ of function $f$. In $\hat{f}(S)$ we restrict the possible values for $z_i$, for every action $i \in [n]$, to $z_i \in \{0, \tau, 2\tau, 3\tau,\ldots,1\}$, where $\tau = 1/\lceil k/\delta\rceil$. This restriction decreases the optimal value by at most a factor of $\delta$, i.e., $\hat{f}(S) \ge (1-\delta)f(S)$ for every subset $S \subseteq [n-1]$. We then construct a knapsack-style FPTAS to find, for any constant $\delta > 0$, a subset $S$ such that $\hat{f}(S) \ge (1-\delta) \hat{f}(S^*) \ge (1-\delta)^2 f(S^*)$ in polynomial time, where $S^* \subseteq [n-1]$ is the set of $k-1$ actions maximizing $f$. Using $\delta = \varepsilon/2$ then yields $\hat{f}(S) \ge (1-\varepsilon) f(S^*)$. By submodularity, $f(S^*) \ge \frac{k-1}{k} \cdot f(S^*_k)$, and, hence, $f(S^*) \ge \frac{k-1}{k} \cdot f(K) \ge \frac{k-1}{k} \cdot u_\sender(\varphi^*)$.

The following proposition summarizes the main insight from this section.
\begin{proposition}
  \label{prop:FPTAS}
  For every $k \ge 2$ and every constant $\varepsilon > 0$, there is a polynomial-time algorithm to compute a subset $S$ of $k-1$ actions such that
    \[  f(S) \ge (1-\varepsilon) \cdot \left(1-\frac 1k\right) \cdot u_\sender(\varphi^*)\enspace.  \]
\end{proposition}
Combining the algorithm for selection of $S$ with ComputeSignal, we obtain a signaling scheme that we term the Improved Independent Scheme. Proposition~\ref{prop:FPTAS} together with Lemmas~\ref{lem:signalApx} and~\ref{lem:signalPersuasive} imply Theorem~\ref{thm:independentFPTAS}.

Let us now describe the algorithm and the guarantee in Proposition~\ref{prop:FPTAS} in more detail. We first apply a discretization, for which we need to solve LP~\eqref{eqn:LP-submodular2} a total of at most $O(nk/\varepsilon)$ times. The subsequent FPTAS procedure needs $O(n^2 k^6/\varepsilon^3)$ time which, arguably, seems rather high. Our goal here was to simplify the exposition and the analysis of the FPTAS as much as possible. It is an interesting direction for future work to improve the running time in terms of the dependence on $k$ and $\varepsilon$.

\paragraph{Discretization}
For approximating $f$, we consider approximating the function $\hat{f}$. The definition of $\hat{f}$ is the same as for $f$ in~\eqref{eqn:f}, where we add a discretization constraint that $z_i \in \{0,\tau,2\tau,3\tau,\ldots,\frac{\tau-1}{\tau},1\}$ with $\tau = 1/\lceil k/\delta \rceil$.

\begin{lemma}
    \label{lem:discrete}
    Consider the subset $S^* \subseteq [n-1]$ that maximizes $f(S^*)$. It holds that $\hat{f}(S^*) \ge (1-\delta)f(S^*)$.
\end{lemma}

\begin{proof}
    Since \eqref{eqn:f} is a packing problem, we can assume w.l.o.g.\ that $|S^*| = k-1$. We denote by $\vecz^*$ the optimal solution for $f(S^*)$ in \eqref{eqn:f}. For $z'_i = (1-\delta)z^*_i$, concavity and monotonicity of $g_i$ implies $g_i(z'_i) \ge (1-\delta) g_i(z^*_i)$ for every $i \in S^* \cup \{n\}$. Observe that $\sum_{i \in S^* \cup \{n\}} z'_i \le (1-\delta)$ since $\vecz^*$ is a feasible solution. We round $z'_i$ up to the next multiple of $\tau$, i.e., $\hat{z}_i = \tau \cdot \lceil \frac{z'_i}{\tau} \rceil$. Then
    \[
    \sum_{i \in S^* \cup \{n\}} \hat{z}_i \le \sum_{i \in S^* \cup \{n\}} z'_i + \tau \le (1-\delta) + k \cdot \frac{1}{\lceil k/\delta \rceil} \le 1\enspace.
    \]
    Now $\bm{\hat{z}}$ is a feasible solution for the optimization problem of $\hat{f}(S)$, so
    \[
        \hat{f}(S) \ge \sum_{i \in S^* \cup \{n\}} g_i(\hat{z}_i) \ge \sum_{i \in S^* \cup \{n\}} g_i(z'_i) \ge (1-\delta) \sum_{i \in S^* \cup \{n\}} g_i(z^*_i) = (1-\delta) f(S^*)\enspace.
    \]
\end{proof}

We rephrase the optimization problem of $\hat{f}(S)$ as having $1/\tau$ many \emph{particles} that can be assigned to the actions $S \cup \{n\}$. The $\ell$-th particle assigned to action $i$ has marginal profit $m_i^{\ell} = g_i(\ell \tau) - g_i((\ell-1) \tau)$. For every action $i$, the marginal profit of the $\ell$-th assigned particle is $m_i^{\ell} \ge 0$ and $m_i^{\ell+1} \le m_i^{\ell}$, for all $\ell \ge 1$. Clearly, the optimal solution for $\hat{f}(S)$ can be computed by a simple greedy algorithm: Assign the $1/\tau$ particles to actions $S \cup \{n\}$ in non-increasing order of marginal profit. Consider the set $\hat{S}^*$ that optimizes $\hat{f}(S)$ over all subsets $S$ of size at most $k-1$. Let $m^*$ be the profit of the last particle assigned by the greedy algorithm to any action in $\hat{S}^* \cup \{n\}$.

Our main idea in the FPTAS is to guess $m^*$. Put differently, we run the algorithm discussed in the following for all marginal profits from all particles of all functions $g_i$, $i \in [n]$. Since only $1/\tau$ particles must be considered for any action, we have at most $n/\tau = O(nk/\delta)$ calls to the algorithm. For the rest of this section, we outline our approach for a given marginal profit value $m$.

Given a value $m$, consider an action $i$. We denote by $\ell_i(m)$ the largest number of a particle with marginal profit strictly larger than $m$. Suppose that $i \in \hat{S}^*$ and $m = m^*$. Then in $\hat{f}(\hat{S}^*)$ we will assign at least $z_i \ge \tau \ell_i(m)$ to action $i$. With foresight, we use the notation $w_i^r(m) = \tau \cdot \ell_i(m)$ and $p_i^r(m) = g_i(\tau \ell_i(m))$. Suppose $i$ has particles $\ell_i(m)+1, \ell_i(m)+2,\ldots,\ell_i(m) + t_i(m)$ with marginal profit $m$, then $\hat{f}$ assigns a mass of $z_i \in [\tau \ell_i(m), \tau(\ell_i(m)+t_i(m))]$. We use the notation $w_i^o(m) = \tau \cdot t_i(m)$ and $p_i^o(m) = m \cdot \tau \cdot t_i(m) = m \cdot w_i^o(m)$. Otherwise, if $g_i$ has no particle with marginal profit $m$, then $z_i = \tau \ell_i(m)$, and we set $w_i^o(m) = p_i^o(m) = 0$. With this notation, we can express $\hat{f}(\hat{S}^*)$ by
\begin{align*}
\hat{f}(\hat{S}^*) &= \sum_{i \in \hat{S}^* \cup \{n\}} g_i(\tau \ell_i(m^*)) + m^* \cdot \sum_{i \in \hat{S}^* \cup \{n\}} (z_i - \tau \ell_i(m^*)) \\
&= \sum_{i \in \hat{S}^* \cup \{n\}} p_i^r(m^*) + m^* \cdot \left(1 - \sum_{i \in \hat{S}^* \cup \{n\}} w_i^r(m^*)\right)\enspace.
\end{align*}
$\hat{f}$ distributes a total of $1/\tau$ particles to $\hat{S}^*$ such that for all functions $g_i$, $i \in \hat{S}^* \cup \{n\}$ we exhaust all particles from these functions with marginal profit strictly larger than $m^*$. The remaining particles achieve a marginal profit of exactly $m^*$. This implies, in particular, that
\begin{equation}
    \label{eqn:optSlopeIneq}
    0\; \le \; 1 - \sum_{i \in \hat{S}^* \cup \{n\}} w_i^r(m^*) \; \le \; \sum_{i \in \hat{S}^* \cup \{n\}} w_i^o(m^*)\enspace.
\end{equation}

\paragraph{Knapsack Problem}
Consider the following integer optimization problem for a given marginal profit value $m$. We strive to find a subset $S$ of at most $k-1$ actions such that $1/\tau$ particles can be assigned with a marginal profit of at least $m$ from actions $i \in S \cup \{n\}$ to maximize the resulting total profit.
\begin{equation}\label{eqn:knapsack}
\begin{array}{llrcll}
\renewcommand{\arraystretch}{0.5}
h(m) = &\mbox{Max.} & \multicolumn{4}{l}{\D \sum_{i = 1}^n y_i p_i^r(m) + \min\left(m  - m\sum_{i=1}^n y_i w_i^r(m), \sum_{i=1}^n y_i p_i^o(m)\right)}\\
& \mbox{s.t.} & \D \sum_{i=1}^n y_i w_i^r(m) &\le& 1 \\
& & \D \sum_{i=1}^{n-1} y_i & \le & k-1 \\
& & y_n & = &1\\
& & y_i & \in & \{0,1\}
\end{array}
\end{equation}
For given $m$, we denote the optimal solution for $h(m)$ by $\vecy^*$ and the action set optimizing $h(m)$ by $S^*_m = \{ i \mid y^*_i = 1, i \neq n\}$.
\begin{lemma}
\label{lem:h}
  For every marginal profit $m$, the following holds:
  \begin{enumerate}[\hspace{0.3cm}(a)]
  \item If $h(m)$ is feasible, then $h(m) \le \hat{f}(\hat{S}^*)$.
  \item If $m = m^*$, then $h(m^*)$ is feasible and $h(m^*) = \hat{f}(\hat{S}^*)$.
  \item If $h(m)$ is infeasible, then $m \neq m^*$.
  \end{enumerate}
\end{lemma}
\begin{proof}
In $h$ we sum the value from each action $i \in S^*_s$ for the required assignment of particles to arrive at marginal profit $m$, and then use the remaining particles to generate additional value at a rate of $m$. Consider any marginal profit $m$ and a feasible solution $\vecy$ for $h(m)$ with action set $S = \{i \mid y_i = 1, i\neq n\}$. If $S$ satisfies \eqref{eqn:optSlopeIneq}, then $h(m) = \hat{f}(S)$, since $h$ correctly captures the greedy algorithm to assign particles to $g_i$ in non-increasing order of marginal profit. However, there might be values $m$ and solutions $\vecy$, such that for the corresponding set $S \cup \{n\}$ of actions it is impossible to find a total of $1/\tau$ particles with marginal profit at least $m$. Clearly, if this happens, then
\[
   \sum_{i \in S \cup \{n\}} w_i^r(m) + \sum_{i \in S \cup \{n\}} w_i^o(m) < 1
\]
This implies, in particular, that either $m \neq m^*$ or $S \neq \hat{S}^*$, since otherwise we would violate~\eqref{eqn:optSlopeIneq}. Moreover, in $h$ the set $S$ only yields a value of
\begin{align*}
   &\sum_{i \in S \cup \{n\}} p_i^r(m) + \min\left(m - m\sum_{i=1}^n w_i^r(m), \sum_{i \in S \cup \{n\}} p_i^o(m)\right) = \sum_{i \in S \cup \{n\}} p_i^r(m) + \sum_{i \in S \cup \{n\}} p_i^o(m)\enspace,
\end{align*}
i.e., it only sums up the value generated by particles with marginal profit at least $m$. In contrast, in $\hat{f}(S)$ we would continue the greedy algorithm and assign particles beyond the ones with marginal profit at least $m$. This holds in particular for $S = S^*_m$, so $h(m) \le \hat{f}(S^*_m)$. Since $\hat{f}(S^*_m) \le \hat{f}(\hat{S}^*)$, this proves (b).

It is straightforward to verify that for $m^*$ and the optimal set $\hat{S}^*$ the conditions in~\eqref{eqn:optSlopeIneq} guarantee that $h(m^*)$ is feasible. Moreover, \eqref{eqn:optSlopeIneq} implies that in the objective function
\begin{align*}
  \min\left(m^* - m^* \sum_{i \in \hat{S}^* \cup \{n\}} w_i^r(m^*), \sum_{i \in \hat{S}^* \cup \{n\}} p_i^o(m^*) \right) = m^* \left(1 - \sum_{i \in \hat{S}^* \cup \{n\}} w_i^r(m^*)\right)\enspace.
\end{align*}
This implies that $h(m^*) = f(\hat{S}^*)$, and (a) follows.

If $h(m)$ is infeasible, then for every subset $S \subseteq [n-1]$ with $|S| \le k-1$ actions we have
\[ \sum_{i \in S \cup \{n\}} w_i^r(m) > 1\enspace.\]
Then $m \neq m^*$ since~\eqref{eqn:optSlopeIneq} is violated. This proves (c).
\end{proof}

\paragraph{Dynamic Program}
As a consequence of Lemma~\ref{lem:h}, in order to compute an approximation to $\hat{f}(\hat{S}^*)$ we focus on approximating $h(m)$ in~\eqref{eqn:knapsack} for every given value $m$. For convenience, we use a knapsack terminology. There is an \emph{required item} for action $i$ with size $w_i^r(m)$ and profit $p_i^r(m)$. In addition, there is an \emph{optional item} with size $w_i^o(m)$ and profit $p_i^o(m)$. The constraints in~\eqref{eqn:knapsack} (with the exception of the trivial constraint $y_n = 1$) exactly represent the constraint set of the 1.5-dimensional knapsack problem~\cite[Section 9.7]{KellererPP13}.

The objective function can be interpreted as follows. Upon packing a required item of action $i$ into the knapsack, we also allow to fill the remaining space in the knapsack with (any fraction of) the optional item of $i$. Note that all optional items correspond to particles with marginal profit $m$. Optional items can be removed to free space for required items of other actions. Since required items correspond to particles with marginal profit larger than $m$, they generate more value per unit of size they occupy in the knapsack. Hence, adding required items (as long as the constraint set allows it) and removing (parts of) optional ones is always desirable.

For every given $m$, we now describe an FPTAS to approximate the optimal solution of~\eqref{eqn:knapsack} by $(1-\delta)$ in polynomial time, for every constant $\delta > 0$. The approach resembles the standard dynamic programming approach for the knapsack problem. We assume w.l.o.g.\ that all required items fit into the knapsack, i.e., $w_i^r(m) \le 1$ for all $i \in [n-1]$, since otherwise we can drop the action from consideration.

Consider $p_{max}(m) = \max \{ p_i^r(m), \min(m, p_i^o(m)) \mid i \in [n]\}$, and assume $\kappa = \frac{\delta\cdot p_{max}(m)}{2k}$. We consider the adjusted profits $\bar{p}_i^r = \lfloor p_i^r(m) / \kappa \rfloor$ and $\bar{p}_i^o = \lfloor p_i^o(m) / \kappa \rfloor$. Our dynamic programming table is given by $A(i,j,\bar{p}^r,\bar{p}^o)$ with the interpretation that for this entry we consider a subset of solutions of the following form: (1) the packed required items are from actions $\{1,\ldots,i,n\}$, (2) we pack the required items of action $n$ and exactly $j$ of the remaining actions, (3) the packed required items have a total adjusted profit of $\bar{p}^r$, (4) the adjusted profit of optional items corresponding to packed required items sums to $\bar{p}^o$. For each entry $A(i,j,\bar{p}^r,\bar{p}^o)$ we store the minimum total size of required items of any solution that fulfills the conditions of this entry. The number of possible table entries is $O(n \cdot k^5/ \delta^2)$, which is a polynomial number in $n$ and $k$. We initialize all entries with $\infty$. Then the base cases of the recursion are
\begin{align*}
  &A(0,0,\bar{p}_n^r, \bar{p}_n^o) = w_n^r, \hspace{0.5cm} \text{ and }\\
  &A(0,0,x,y) = \infty \quad \text{ for every } x,y \in \{0,1,\ldots,k \cdot \lfloor k/\delta \rfloor\}, (x,y) \neq (\bar{p}_n^r, \bar{p}_n^o).
\end{align*}
We fill the table in increasing order of the parameters by setting
\[
    A(i,j,\bar{p}^r, \bar{p}^o) = \min\{ A(i-1,j,\bar{p}^r, \bar{p}^o), w_i^r + A(i-1,j-1,\bar{p}^r - \bar{p}_i^r, \bar{p}^o - \bar{p}_i^o)\}\enspace,
\]
where we assume the entry is $\infty$ whenever the arguments become negative. Clearly, this recursion allows to fill the table in time linear in the size of the table. As in the standard knapsack problem, the recursion simply distinguishes between packing the required item of action $i$ into the knapsack or not.

The rationale behind this approach is as follows. Consider the set of solutions represented by $A(i,j,\bar{p}^r,\bar{p}^o)$. Clearly when we have $\bar{p}^r$ adjusted profit from packed required items and a potential adjusted profit of $\bar{p}^o$ from optional items, the best solution is one that minimizes the size of packed required items to allow for a maximum portion of optional items to be included into the knapsack.

After completing the table, we consider all entries with $A(i,j,\bar{p}^r, \bar{p}^o) \le 1$, since these entries correspond to a feasible solution. From each of these entries, we pick the one that maximizes the adjusted profit $\kappa \cdot \bar{p}^r + \min(m - m\cdot A(i,j,\bar{p}^r, \bar{p}^o), \kappa \cdot \bar{p}^o)$.

\paragraph{Approximation Ratio}
Consider the adjusted profit of the optimal solution $S^*_m$, which is
\begin{align*}
\sum_{i \in S^*_m \cup \{n\}} \kappa \bar{p}_i^r +  \min\left(m - m\cdot \sum_{i \in S^*_m \cup \{n\}} w_i^r, \sum_{i \in S^*_m \cup \{n\}} \kappa \bar{p}_i^o\right) \; \ge \;  h(m) -2k \kappa \; = \; h(m) - \delta p_{max}
\end{align*}
If $p_{max}$ is attained for a profit of a required item $p_i^r$, then consider packing only the required item $i$. This is a feasible solution since $w_i^r \le 1$. Otherwise, suppose $p_{max}$ is attained for an entry $\min(m,p_i^o(m))$. We use $\min(m,p_i^o(m))$ in the definition of $p_{max}$, since the optional item is not assumed to fit into the knapsack completely, and $m \cdot 1$ is the profit of a knapsack filled completely with any set of (parts of) optional items. Now suppose we pack only the optional item of $i$ (or parts of it until the knapsack is full). Then pack the required item of $i$, thereby possibly replacing parts of the optional item. This is a feasible solution since $w_i^r(m) \le 1$. The replacement increases the profit over $\min(m,p_i^o(m))$. Overall, these observations imply $h(m) \ge p_{max}$.

The dynamic program computes a solution $S'$ with the best adjusted profit. The profit of $S'$ is more than the adjusted profit, which is more than the adjusted profit of $S^*_m$, which is more than $h(m) - \delta p_{max}$. Since $h(m) \ge p_{\max}$, the profit of $S'$ is at least $(1-\delta) \cdot h(m)$.

Since we run the dynamic program for all marginal profits of particles, the best solution $S$ that is found overall has value $f(S) \ge \hat{f}(S) \ge (1-\delta) h(m^*) = (1-\delta) \hat{f}(\hat{S}^*) \ge (1-\delta)^2 f(S^*) = (1-\varepsilon) f(S^*)$ due to Lemmas~\ref{lem:discrete} and~\ref{lem:h}.

\subsection{Beyond $\rv_E$-Optimality}

Let us briefly observe that our approach does not easily translate to independent instances without $\rv_E$-optimality. Consider the following example. There are $n = 2$ actions and $k = 2$ signals. Action 1 has deterministic type $\Theta_1 = \{\theta_{11}\}$ with $(\sv_{11},\rv_{11}) = (1,0)$. Action 2 has types $\Theta_2 = \{\theta_{21}, \theta_{22}\}$ with $(\sv_{21},\rv_{21}) = (0,1)$ and $(\sv_{22},\rv_{22}) = (0,0)$, and $q_{21} = q_{22} = 1/2$. Note $\rv_E = 1/2$ for action 2.

The optimal scheme $\varphi^*$ recommends action 1 in state $(\theta_{11},\theta_{22})$ and action 2 in state $(\theta_{11},\theta_{21})$. In the former case, $\calR$ has conditional expectation of 0 for each of the actions, so action 1 is a best response. In the latter case, the recommended action is optimal for $\calR$. The expected utility for $\calS$ in $\varphi^*$ is 1/2.

Instead, suppose we solve LP~\eqref{eqn:f}. Since the constraints in~\eqref{eqn:LP-submodular2} require a conditional expectation of $\rv_E$ for every signal, the optimal solution is $x_{21}^* = x_{22}^* = 1/2$, and thus $g_1(z^*) = g_2(z^*) = 0$. Hence, the optimal value of the LP is 0. Clearly, the optimal scheme does not give rise to a feasible LP solution, and the optimal LP-value does not upper bound the expected utility of $\varphi^*$ for $\calS$.

More fundamentally, any positive value for $\calS$ results from $\calR$ taking action 1, which in turn must be inherently correlated with the state of action 2. This correlation is not sufficiently reflected in the LP or the algorithms above, which exploit independence conditions. Obtaining a constant-factor approximation for general independent instances in polynomial time is an interesting open problem.

\section{Approximation by Restricted Signals}
\label{sec:apxNbyK}

Let $\opt_k$ denote the expected sender utility of the optimal scheme with $k$ signals. We quantify the performance loss against a case when the sender has (at least) $n$ signals available and achieves $\opt_n$.

\subsection{Symmetric Instances}

We define the Imitation Scheme $\varphi_{Imi}$ for any symmetric instance with $n$ actions and $k$ signals. It first runs an optimal symmetric scheme $\phi^*_n$ for $n$ signals. Let $i$ be the action chosen by $\phi^*_n$. If $i \in [k]$, we signal action $i$; otherwise, we signal any action chosen uniformly at random from $[k]$.

The running time of $\varphi_{Imi}$ is determined by the running time to implement an optimal symmetric scheme for $n$ signals. In particular, such a scheme can be computed with the Slope-Algorithm, so an efficient probability oracle is sufficient for polynomial running time of $\varphi_{Imi}$. We now show that $\varphi_{Imi}$ provides a tight approximation ratio in terms of $\opt_n$.

\begin{proposition}
	The Imitation Scheme is symmetric, direct, and persuasive in symmetric instances. For every $k\ge 2$ it holds $u_{\sender}(\phi_{Imi}) \ge k/n\cdot \opt_n$. There exists a random-order instance such that $\opt_k \le k/n \cdot \opt_n$.
\end{proposition}

\begin{proof}
    We first prove the result for the Imitation Scheme. The optimal scheme $\phi^*_n$ is symmetric. If $\phi_{Imi}$ deviates from the recommendation of $\phi^*_n$, it recommends a uniform random action in $[k]$. Hence, $\phi_{Imi}$ is also symmetric.

    Conditioned on action $i$ being recommended by $\phi_{Imi}$, the type distribution of action $i$ is $\dist_{yes}$ from $\phi^*_n$ with probability $k/n$ or $\dist_{no}$ from $\phi^*_n$ with probability $(n-k)/n$. If an action $i \in [k]$ is not recommended, the type distribution is $\dist_{no}$ from $\phi^*_n$, no matter which action $j \in [k]$ is recommended. Let $\rv_{yes}$ and $\rv_{no}$ be the expected utility of $\receiver$ in $\dist_{yes}$ and $\dist_{no}$ from $\phi^*_n$, respectively. In $\phi_{Imi}$, the expected utility for $\receiver$ for any given action $i \in [k]$ must satisfy
    \[
    \frac{1}{k} \left(\frac{k}{n} \cdot \rv_{yes} + \frac{n-k}{n} \cdot \rv_{no} \right) + \frac{k-1}{k} \cdot \rv_{no} = \rv_E\enspace.
    \]
    Since $\rv_{no} \le \rv_E$, the expected utility for $\receiver$ when following a recommended action in $\phi_{Imi}$ can be bounded by $\frac{k}{n} \cdot \rv_{yes} + \frac{n-k}{n} \cdot \rv_{no} \ge \rv_E$. By Lemma~\ref{lem:symPersuasive} we see that $\phi_{Imi}$ is persuasive.

    The optimal scheme $\phi^*_n$ is symmetric and recommends each action with probability $1/n$. With probability $k/n$, $\phi_{Imi}$ recommends the same action as $\phi^*_n$, so $u_{\sender}(\phi_{Imi}) \ge k/n\cdot \opt_n$.

    For the upper bound on $\opt_k$, we consider an instance from the random-order scenario. There are $n$ types. $\theta_1$ has utility pair (1,1), all $n-1$ remaining types have utility pair $(0,0)$. Obviously, $\opt_n = 1$, the sender gives a signal for the action with type 1. With $k$ signals, there is an optimal scheme that recommends only the first $k$ actions. With probability $k/n$, type 1 is among those $k$ actions. Otherwise, type 1 cannot be recommended. Hence, $\opt_k \le k/n$, which completes the proof.
\end{proof}

\subsection{Independent Instances}

The first lemma shows that there are independent (and symmetric) instances such that the best approximation ratio is in $O(k/n)$.

\begin{lemma}
   There exists an IID instance such that $\opt_k \le \frac{e}{e-1} \cdot k/n \cdot \opt_n$.
\end{lemma}

\begin{proof}
In the distribution for every action there is a good type $\state_1$ with utility pair $(1,1)$ and $q_{\state_1} = 1/n$, and a bad type $\state_0$ with utility pair $(0,0)$ and $q_{\state_0} = 1-1/n$. Clearly, the optimal mechanism is to signal an action with $\state_1$ whenever it exists (within the first $k$ actions). This yields a ratio of
\[
\frac{\opt_k}{\opt_n} = \frac{1 - \left(1- \frac 1n \right)^k}{1 - \left(1- \frac 1n \right)^n} = \frac{k/n - \sum_{i=2}^k {k \choose i} \left(\frac{-1}{n}\right)^i}{1 - \sum_{i=2}^n {n \choose i} \left(\frac{-1}{n}\right)^i} 
\]
%

This ratio is at most $e/(e-1) \cdot k/n$, for all $k \in \{2,\ldots,n\}$, where $e/(e-1) \approx 1.58$. To see this, observe
\[
\opt_k \ge \frac{k}{k+1} \cdot \opt_{k+1}
\]
since the left-hand side is a lower bound on the sender utility when using $\phi_{Imi}$ for $k$ signals based on the optimal IID-scheme for $k+1$ signals. Hence,
\[
\frac{\opt_k}{k} \ge \frac{\opt_{k+1}}{k+1}
\]
Therefore, for all $k=2,\ldots,n$ we have $\opt_k/k \le \opt_1/1 = \opt_1 = 1/n$, i.e.,
\[
\frac{\opt_k}{\opt_n} \Big/ \frac{k}{n} = \frac{n}{\opt_n} \cdot \frac{\opt_k}{k} \le \frac{1}{\opt_n} = \frac{1}{1 - \left(1- \frac 1n \right)^n}\enspace.
\]
The factor is monotone in $n$ and grows to $1/(1-1/e) = e/(e-1)$.
\end{proof}

To provide an asymptotically tight bound of $\Omega(k/n)$, we consider the Independent-Imitation Scheme. Similar to the schemes in Sections~\ref{sec:constantFactor} and~\ref{sec:FPTAS} above, it consists of the two steps of (a) choosing a suitable subset of actions and (b) computing a good direct signaling scheme for the chosen subset of actions. In the Independent-Imitation Scheme we use ActionsReduce (Algorithm~\ref{algo:reduceActions}) for step (a), and ComputeSignal (Algorithm~\ref{algo:signal}) for step (b) as above. Since the main computational step in both algorithms is to solve a single linear program, the scheme can be implemented in polynomial time.

\begin{algorithm}[t]
   \DontPrintSemicolon
	\KwIn{Type sets $\Theta_1, \ldots, \Theta_n$ and distributions $q_1,\ldots,q_n$, s.t.\ $\sum_j q_{n,j}\rv_{nj} = \rv_E$ and $\sum_j q_{n,j}\sv_{nj} = \max_{i \in [n] \, : \, \sum_j q_{i,j}\rv_{ij} = \rv_E} \sum_j q_{i,j} \sv_{ij}$, parameter $2 \le k \le n$}

    Compute $f([n-1])$.\;
    For every $i \in [n]$, let $z^*_i$ be the values of the optimal solution in $f([n-1])$.\;
    Let $S$ be the set of the $k-1$ actions from $[n-1]$ with largest values $g_i(z^*_i)$.\;
    \Return{$S$}

   \caption{\label{algo:reduceActions} ActionsReduce}
\end{algorithm}
\begin{theorem}
   The Independent-Imitation Scheme is direct and persuasive for independent instances with $k$ signals. It can be implemented in time polynomial in the input size. For every $k \ge 2$,
   \[ u_{\sender}(\phi_{ImiIS}) \ge \left(1-\left(1-\frac 1k\right)^k\right) \cdot \left(1 - \frac 1k\right) \cdot \frac{k}{n} \cdot \opt_n \enspace.\]
\end{theorem}

\begin{proof}
   Following \eqref{eqn:f} above we observed that $f(S) \ge u_{\sender}(\varphi^*_{S \cup \{n\}})$, so in particular $f([n-1]) \ge \opt_n$. For every action $i \in [n-1]$ and every type $j \in \Theta_i$, let $z^*_i$ and $x^*_{ij}$ be the values of the optimal LP solution for $f([n-1])$. It is straightforward to verify that for every subset $S$, the values $(z^*_i)_{i \in S \cup \{n\}}$ and $(x^*_{ij})_{i \in S \cup \{n\}, j \in \Theta_i}$ constitute a feasible solution for the LP when optimizing $f(S)$. Since ActionsReduce chooses $S$ to contain the $k-1$ actions with largest $g_i(z^*_i)$,
   \[
       f(S) \; \ge \; \sum_{i \in S \cup \{n\}} g_i(z^*_i) \; \ge \; \frac{k-1}{n} \cdot f([n-1]) \; \ge \; \left(1 - \frac 1k \right) \cdot \frac{k}{n} \cdot \opt_n
   \]
   The approximation ratio now follows using Lemma~\ref{lem:signalApx}. By Lemma~\ref{lem:signalPersuasive}, the resulting signaling scheme is direct and persuasive.
\end{proof}

\bibliographystyle{plain}
\bibliography{literature,martin}

\clearpage
\appendix

\section{Proof of Lemma~\ref{lem:directPersuasive}}
\label{app:directPersuasive}

\begin{proof}
    The first statement follows from a simple revelation-principle-style argument. Consider any signaling scheme $\varphi$. Given any signal $\sigma$, we can assume $\receiver$ chooses one action that maximizes the conditional expectation of her utility. Suppose for two signals $\sigma, \sigma'$, $\receiver$ chooses the same action. Then $\sender$ can simply drop $\sigma'$ and issue $\sigma$ every time it issued $\sigma'$, thereby achieving the same behavior of $\receiver$. Thus, each of the $k$ signals can be assumed to correspond to a distinct choice of action of $\receiver$, which maximizes the conditional expectation of the utility of $\receiver$. Hence, we can equivalently assume that $\sender$ uses the signal to issue a direct recommendation for an action such that $\receiver$ wants to follow the recommendation.

    For the second statement, symmetry in the instance allows to restrict attention to the first $k$ actions. Consider an optimal direct and persuasive scheme $\varphi$ that recommends actions from a size-$k$-subset $K \subseteq [n]$. Permute the labels of all actions in $\varphi$ (and w.l.o.g.\ tie-breaking rule of $\receiver$) such that it recommends actions from $[k]$. Denote the permuted $\varphi$ by $\varphi'$. Since the distribution over states of nature is symmetric, it is invariant to permutation of action labels. Hence, applying $\varphi'$ yields the same conditional expectations for the utility of $\receiver$ for the actions in $[k]$ as $\varphi$ yields for $K$. Thus, $\varphi'$ is direct and persuasive with recommendations from $[k]$ and $u_\sender(\varphi) = u_\sender(\varphi')$.
\end{proof}

\section{Efficient Probability Oracles}
\label{app:efficientProbability}
We divide the proof of Theorem~\ref{thm:symmetricPolyTime} into two subsections for the prophet-secretary and the $d$-random-order scenarios.
\subsection{Prophet-Secretary}
\label{sec:prophetSec}

Consider the prophet-secretary scenario, in which we have $n$ probability distributions over type spaces $\Theta^1,\ldots,\Theta^n$, respectively. For simplicity, we reverse the generation process of the state of nature $\vecState$: First, permute the $n$ distributions in uniform random order, then draw a single type from each distribution independently.

We denote by $q^i_\theta$ the probability that type $\theta \in \Theta^i$ is drawn from distribution $i$. For the $n$ type spaces, we assume w.l.o.g.\ that they are mutually disjoint. In addition, we assume for simplicity that types are in general position, i.e., there are no more than two distinct types on any given straight line. We discuss in the end how our observations can be adapted when this assumption does not hold.

Overall, the representation size of the input is at least linear in $n$, $\max_i |\Theta^i|$, and $\max_{i,\theta} \log 1/q^i_\theta$. For a polynomial-time probability oracle, we have to implement two classes of queries in time polynomial in the aforementioned quantities:
\begin{enumerate}[a)]
	\item Given a pair of types $a$ and $b$, return the probability $p_{ab}$ that $\overline{ab}$ is in the Pareto-frontier of the type set $C$ of the first $k$ actions.
	\item Given a type $c$ and a slope $s$, return the probability $p_c^{(s)}$ that $c$ is the unique point that corresponds to slope $s$ on the Pareto-frontier of the type set $C$ of the first $k$ actions.
\end{enumerate}

\paragraph{Class a)}
If the two types are from the same distribution, then $p_{ab} = 0$. Otherwise, let $i_a$ and $i_b$ be such that $a \in \Theta^{i_a}$ and $b \in \Theta^{i_b}$. For each distribution $i \neq i_a,i_b$ we consider every type $c \in \Theta^i$. If $c$ lies above the line through $a$ and $b$ and is included in $C$, then $c$ lifts the Pareto frontier above $\overline{ab}$, and the segment would vanish from the Pareto-frontier. Thus, if $c$ lies above the line through $a$ and $b$, then $c$ must not be in $C$. Otherwise, $c$ is an \emph{allowed type}. We denote by $\Theta^i_{ab}$ the set of allowed types of distribution $i$, and by $q^i_{ab} = \sum_{c \in \Theta_{ab}^i} q_c^i$ the probability to draw an allowed type in distribution $i$. Clearly, these probabilities can be determined in time linear in the total number of types.

Now in order to have $\overline{ab}$ on the Pareto-frontier, it must be the case that (1) distributions $i_a$ and $i_b$ are permuted to the first $k$ actions, (2) $a$ and $b$ are drawn from distributions $i_a$ and $i_b$, respectively, and (3) for every other distribution $i$ permuted to the first $k$ actions, we draw an allowed type. The probability for (1) is $\frac{k}{n} \cdot \frac{k-1}{n-1}$, the probability for (2) is $q^{i_a}_a \cdot q^{i_b}_b$. To compute the probability of (3), we consider every subset $A \subseteq \{1,\ldots,n\} \setminus \{i_a,i_b\}$ of $|A| = k-2$ distributions and compute the probability that from every distribution of $A$ we draw an allowed type. Overall,
\[
p_{ab} = \frac{k}{n} \cdot \frac{k-1}{n-1} \cdot q^{i_a}_a \cdot q^{i_b}_b \cdot \frac{1}{{n-2 \choose k-2}} \cdot \sum_{\substack{A \subseteq \{1,\ldots,n\} \setminus \{i_a,i_b\}\\ |A| = k-2}} \prod_{i \in A} q^i_{ab}\enspace.
\]
To compute the last term, we need to compute the sum of products of all $(k-2)$-size subsets of $n-2$ numbers. This can done in time $O(nk)$ using a dynamic program.

\paragraph{Class b)}
Let $i_c$ be such that $c \in \Theta^{i_c}$. For each distribution $i \neq i_c$ we again consider every type $d \in \Theta^i$. $c$ shall be the unique point corresponding to slope $s$ on the Pareto frontier, so there must not be any type on or above the line going through $c$ with slope $s$. Hence, all types that remain strictly below this line are \emph{allowed types}.  We denote by $\Theta^i_c$ the set of allowed types of distribution $i$, and by $q^i_c = \sum_{d \in \Theta_c^i} q_d^i$ the probability to draw an allowed type in distribution $i$. These probabilities can be determined in time linear in the total number of types.

For $c$ to be the unique point that corresponds to $s$ on the Pareto-frontier, it must be the case that (1) distribution $i_c$ is permuted to the first $k$ actions, (2) $c$ is drawn from distribution $i_c$, and (3) for every other distribution $i$ permuted to the first $k$ actions, we draw an allowed type. The probability for (1) is $\frac{k}{n}$, the probability for (2) is $q^{i_c}_c$. To compute the probability of (3), we consider every subset $A \subseteq \{1,\ldots,n\} \setminus \{i_c\}$ of $|A| = k-1$ distributions and compute the probability that from every distribution of $A$ we draw an allowed type. Overall,
\[
p_c^{(s)} = \frac{k}{n} \cdot q^{i_c}_c \cdot \frac{1}{{n-1 \choose k-1}} \cdot \sum_{\substack{A \subseteq \{1,\ldots,n\} \setminus \{i_c\}\\ |A| = k-1}} \prod_{i \in A} q^i_c\enspace.
\]
Again, the last term can be computed by a dynamic program in time $O(nk)$.

\paragraph{On general position}
When types are not in general position, i.e., there are three or more types on a straight line, the events of them forming a line segment on the Pareto frontier are not disjoint. Hence, the probabilities to set up the LP have to be computed in a slightly different manner.

We ensure that a segment is not counted multiple times by considering, for any given slope, only the longest possible line segment in the Pareto frontier. Hence, the following modification has to be made for queries of class a) when determining the set of allowed types $\Theta^i_{ab}$: All types of $\Theta^i$ that are on the segment $\overline{ab}$ are allowed since this does not prohibit $\overline{ab}$ to be the longest possible line segment. All types of $\Theta^i$ that are on the straight line going through $a$ and $b$ but not on $\overline{ab}$ must not be allowed. With this modification to calculate the probabilities for $p_{ab}$, general position of types is no longer required.

The main insight in this section is summarized in the following proposition.

\begin{proposition}
	For the prophet-secretary scenario we can implement a probability oracle for the Slope-Algorithm in polynomial time.
\end{proposition}

%

\subsection{$d$-Random-Order}
\label{sec:dRandomOrder}

For $d$-random-order instances, we have $d$ type vectors $\vecState^1,\ldots,\vecState^d$ and a distribution over these vectors. We denote by $q_{\vecState^i}$ the probability of $\vecState^i$. Without loss of generality we assume that all $dn$ types in the $d$ vectors are distinct. To generate a state of nature, we draw vector $\vecState^i$ with probability $q_{\vecState^i}$, and then permute the vector uniformly at random. The representation size of the input is linear in $dn$ and $\max_{i} \log 1/q_{\vecState^i}$. For a polynomial-time probability oracle, we again have to implement two classes of queries discussed in the previous section. The running time will be polynomial in the aforementioned quantities. For simplicity, we again assume types are points in general position.

\paragraph{Class a)}
If the two types $a$ and $b$ come from different vectors $\vecState^j$ and $\vecState^{j'}$, then $p_{ab} = 0$. Otherwise, suppose $a$ and $b$ are from $\vecState^j$. Consider each type $c$ from $\vecState^j$ with $c \neq a,b$. If $c$ lies above the line through $a$ and $b$ and is included in $C$, then $c$ lifts the Pareto frontier above $\overline{ab}$, and the segment would vanish from the Pareto-frontier. Thus, if $c$ lies above the line through $a$ and $b$, then $c$ must not be in $C$. Otherwise, $c$ is an \emph{allowed type}. We denote by $A^j_{ab}$ the set of allowed types from vector $\vecState^j$. Clearly, $A^j_{ab}$ can be computed in time linear in $n$.

Now in order to have $\overline{ab}$ on the Pareto-frontier, it must be the case that (1) $\vecState^j$ is drawn from the distribution, (2) $a$ and $b$ are permuted to the first $k$ actions, and (3) every other type from $\vecState^j$ permuted to the first $k$ actions is an allowed type. The probabilities for these events are (1) $q_{\vecState^j}$, (2) $\frac{k}{n} \cdot \frac{k-1}{n-1}$, and (3) ${|A^j_{ab}| \choose k-2} \Big/ {n-2 \choose k-2}$, where we assume that ${|A^j_{ab}| \choose k-2} = 0$ if $|A^j_{ab}| < k-2$. Overall,
\[
p_{ab} = q_{\vecState^j} \cdot \frac{k}{n} \cdot \frac{k-1}{n-1} \cdot \frac{1}{{n-2 \choose k-2}} \cdot {|A^j_{ab}| \choose k-2}\enspace.
\]
Clearly, this expression can be computed in polynomial time for every pair of types $a,b$.

\paragraph{Class b)}
Let $c$ be a type from vector $\vecState^j$. Consider each type $d$ from $\vecState^j$ with $d \neq c$. Type $c$ shall be the unique point corresponding to slope $s$ on the Pareto frontier, so $d$ must not be on or above the line going through $c$ with slope $s$. If $d$ remains strictly below this line, it is an \emph{allowed type}. We denote by $A^j_c$ the set of allowed types from $\vecState^j$. Clearly, $A^j_c$ can be computed in time linear in $n$.

For $c$ to be the unique point that corresponds to $s$ on the Pareto-frontier, it must be the case that (1) $\vecState^j$ is drawn from the distribution, (2) $c$ is permuted to the first $k$ actions, and (3) every other type from $\vecState^j$ permuted to the first $k$ actions is an allowed type. The probabilities for these events are (1) $q_{\vecState^j}$, (2) $\frac{k}{n}$, and (3) ${|A^j_c| \choose k-1} \Big/ {n-1 \choose k-1}$, where we assume that ${|A^j_{ab}| \choose k-1} = 0$ if $|A^j_{ab}| < k-1$. Overall,
\[
p_c^{(s)} = q_{\vecState^j} \cdot \frac{k}{n} \cdot \frac{1}{{n-1 \choose k-1}} \cdot {|A^j_c| \choose k-1}\enspace.
\]
Again, the expression can be computed in polynomial time for every type $c$.

The adjustments to remove the assumption of general position are the same as for prophet-secretary in the previous section. The main insight in this section is summarized in the following proposition.

\begin{proposition}
	For the $d$-random-order scenario we can implement a probability oracle for the Slope-Algorithm in polynomial time.
\end{proposition}

%

\section{Bicriteria Approximation}
\label{app:bicriteria}

For a symmetric instance with $k$ signals and $n$ actions, consider a \emph{truncation} operation: Remove actions $k+1,\ldots,n$ from consideration and restrict every state of nature $\vecState$ to its first $k$ entries. This yields the \emph{truncated instance} with $k$ signals and $k$ actions. Suppose we apply the characterization from Theorem~\ref{thm:sPareto} and the Slope-Algorithm to compute an optimal scheme in the original instance and the truncated instance. Indeed, it is a straightforward consequence of symmetry that the resulting scheme is the same.

\begin{proposition}
   For symmetric instances with $k$ signals and $n$ actions, there is an optimal scheme that is an optimal scheme for the truncated instance with $k$ signals and $k$ actions, and vice versa.
\end{proposition}

We show that one can apply algorithms to the truncated instance and obtain similar results for the scenario with $k < n$ actions. By truncating the instance, we return to the standard scenario of Bayesian persuasion with $n = k$ actions and signals.

In particular, for the IID scenario, truncation yields an instance where we draw from the same underlying distribution simply for $k$ instead of $n$ actions. Hence, an optimal scheme with $n$ IID actions and $k$ signals is an optimal scheme for $k$ IID actions and $k$ signals. Instead of using the Slope-Algorithm, it can also be obtained by solving a single LP of polynomial size~\cite{DughmiX16}.

More generally, applying a Monte-Carlo sampling approach we obtain a bicriteria approximation when having black-box access to the prior over states of nature. In the following, we assume that all utility values are in $\rv(\state_i), \sv(\state_i) \in [-1,1]$. We assume $\sender$ has black-box oracle access to the prior, i.e., she can draw states of nature as samples from the distribution.

Let $\varphi^*$ be an optimal direct and persuasive scheme. Given any parameter $\varepsilon > 0$, a direct scheme $\varphi$ is $\varepsilon$-persuasive if $\Ex{\rv(\state_i) \mid \sigma = i} \ge \Ex{\rv(\state_j) \mid \sigma = i} - \varepsilon$ for all actions $j \in [n]$. A direct scheme is \emph{$\varepsilon$-optimal} if $u_\sender(\varphi) \ge u_\sender(\varphi^*) - \varepsilon$, where for $u_\sender(\varphi)$ we assume that $\receiver$ follows the recommendation. An $\varepsilon$-persuasive and $\varepsilon$-optimal scheme gives both players a guarantee that their expected utility is at most an (additive) $\varepsilon$ away from a utility benchmark. For $\receiver$ the benchmark is the utility of the best action given $\varphi$, for $\sender$ it is the utility obtained by an optimal persuasive scheme.

The main result of this section is that the bicriteria FPTAS from~\cite{DughmiX16} can be applied to the truncated instance.

\begin{corollary}
  \label{cor:bicriteria}
  In symmetric instances with $k \le n$ signals, utility values in $[-1,1]$, and black-box oracle access to the distribution over states of nature, an $\varepsilon$-persuasive and $\varepsilon$-optimal scheme can be computed in time polynomial in $n$ and $1/\varepsilon$, for every $\varepsilon > 0$.
\end{corollary}

\begin{proof}
  We apply the bicriteria FPTAS from~\cite{DughmiX16} to the truncated instance, which implies the result for the truncated instance. We now argue that the guarantees of $\varepsilon$-optimal and $\varepsilon$-persuasive also apply in the original instance. Since we observed above that there is a scheme that is optimal in both the truncated and original instances, $\varepsilon$-optimality is immediate. It remains to show $\varepsilon$-persuasiveness.

  In addition to the given state of nature, the scheme draws a polynomial number of independent samples from the black-box oracle. It then computes the optimal direct and $\varepsilon$-persuasive scheme for the uniform distribution over the sample set. This is done by solving an LP of polynomial size (see~\cite[Section 5.1]{DughmiX16}). In the solution of the LP, we assume that all ties are broken uniformly at random. Then, if we permute all states of nature in the sample in the same way, the resulting scheme also permutes the signal distributions in the same way. Due to symmetry in the instance, every permutation is equally likely. As a consequence, the resulting scheme is symmetric. We denote by $\rv_{yes}$, $\rv_{no}$, and $\rv_{never}$ the expected utility for $\receiver$ for the distributions of recommended action, non-recommended action in $[k]$, and non-recommended action in $[n] \setminus [k]$, respectively. Note that $\rv_{never} = \rv_E$. Due to symmetry, as in Lemma~\ref{lem:symPersuasive}, we have $\frac{1}{k} \cdot \rv_{yes} + \frac{k-1}{k} \cdot \rv_{no} = \rv_E$, and due to $\varepsilon$-persuasiveness in the truncated instance we know $\rv_{yes} \ge \rv_{no} - \varepsilon$. Combining the two inequalities leads to $\rv_{yes} \ge \frac{k \rv_E - \rv_{yes}}{k-1} - \varepsilon$, which implies $\rv_{yes} \ge \rv_{never} - \frac{k-1}{k} \cdot \varepsilon$.
\end{proof}

\end{document}